\documentclass[journal,9pt]{IEEEtran}
\usepackage{graphicx}
\usepackage{amsmath}
\usepackage{amssymb}
\usepackage{amsfonts}
\usepackage{color}
\usepackage{amsthm}
\usepackage{algorithm}
\usepackage{algpseudocode}
\usepackage{multirow}

\newtheorem{remark}{Remark}

\newtheorem{lemma}{Lemma}
\newtheorem{proposition}{Proposition}

\theoremstyle{definition}
\newtheorem{definition}{Definition}

\def\N{{\mathbb N}}

\def\R{{\mathbb R}}

\title{Distributed Optimisation with Linear Equality and Inequality Constraints using PDMM}
\author{Richard Heusdens \IEEEmembership{Senior Member} and Guoqiang Zhang \IEEEmembership{Member} 

\thanks{R.\ Heusdens is with the Netherlands Defence Academy (NLDA), the Netherlands, and with the Faculty of Electrical Engineering, Mathematics and Computer Science, Delft University of Technology, Delft, the Netherlands (email: r.heusdens@\{mindef.nl,tudelft.nl\}).
}
\thanks{G.\ Zhang is with the University of Exeter, Exeter, United Kingdom (email: g.z.zhang@exeter.ac.uk)}
}

\begin{document}


%\author{Guoqiang Zhang}
%\email{guoqiang.zhang@uts.edu.au}
%\affiliation{%
%  \institution{University Technology of Sydney, Australia}
%  \state{Sydney}
%  \country{Australia}
%}

%\author{Kenta Niwa}
%\email{kenta.niwa.bk@hco.ntt.co.jp}
%\orcid{0000-0002-6911-0238}
%\affiliation{%
%\institution{NTT Communication Science Laboratories}
%  \streetaddress{3-9-11, Midori-Cho}
%  \city{Musashino}
  %\state{Tokyo}
%  \postcode{180-8585}
%  \country{Japan}
%}

%\author{W. Bastiaan Kleijn}
%\email{bastiaan.kleijn@ecs.vuw.ac.nz}
%\affiliation{%
%  \institution{Victoria University of Wellington, New Zealand}
%  \state{Wellington}
%  \country{New Zealand}
%}

%\renewcommand{\shortauthors}{G. Zhang, et al.}

\maketitle

\begin{abstract}
In this paper, we consider the problem of distributed optimisation of a separable convex cost function over a graph, where every edge and node in the graph could carry both linear equality and/or inequality constraints. We show how to modify the primal-dual method of multipliers (PDMM), originally designed for linear equality constraints, such that it can handle inequality constraints as well. %In contrast to most existing algorithms for optimisation with inequality constraints, the proposed algorithm does not need any slack variables. 
The proposed algorithm does not need any slack variables, which is similar to the recent work \cite{He23ADMMIneq} which extends the alternating direction method of multipliers (ADMM) for addressing decomposable optimisation with linear equality and inequality constraints.
Using convex analysis, monotone operator theory and fixed-point theory, we show how to derive the update equations of the modified PDMM algorithm by applying Peaceman-Rachford  splitting to the monotonic inclusion related to the lifted dual problem. To incorporate the inequality constraints, we impose a non-negativity constraint on the associated dual variables. This additional constraint results in the introduction of a reflection operator to model the data exchange in the network, instead of a permutation operator as derived for equality constraint PDMM. Convergence for both synchronous and stochastic update schemes of PDMM are provided. The latter includes asynchronous update schemes and update schemes with transmission losses. Experiments show that PDMM converges notably faster than extended ADMM of \cite{He23ADMMIneq}.
% In addition, we show that a half-averaged version of the proposed algorithm is equivalent to ADMM with inequality constraints.  

%To utilize PDMM for the optimisation problem, we reformulate  each inequality constraint into two equality constraints over an edge by introducing additional slack variables. Theoretical analysis is provided, showing that the saddle point of the so-called augmented primal-dual Lagrangian in PDMM is an optimal solution of the original problem. Extensive experiments are performed, demonstrating that by proper formulation, PDMM can effectively handle not only equality but also inequality constraints.     
\end{abstract}

%\keywords{Distributed optimisation, inequality contraints, ADMM, PDMM}


%\maketitle

\section{Introduction}
In the last decade, distributed optimisation \cite{Boyd11ADMM} has drawn increasing attention due to the demand for either distributed signal processing or massive data processing over a pear-to-pear (P2P) network of ubiquitous devices.  Its basic principle is to first formulate an optimisation problem from the collected or manually allocated data in the devices, and then performing information spreading and fusion across the devices collaboratively and iteratively until reaching a global solution of the optimisation problem. Examples include training a machine learning model, target localisation and tracking, healthcare monitoring, power grid management, and environmental sensing.  In general, the typical challenges faced by distributed optimisation over a network, in particular ad-hoc networks, are the lack of infrastructure, limited connectivity, scalability, data heterogeneity across the network,  data-privacy requirements, and heterogeneous computational resources \cite{ Dimakis10GossipAlg, Li19Fed}. 

Depending on the applications, various methods have been developed for addressing one or more challenges in the considered network. For instance, the work \cite{Boyd06gossip,ust:10} proposed a pairwise gossip method to allow for asynchronous message-exchange in the network, while \cite{ben:10} describes a combination of gossip and geographic routing. In \cite{Iutzeler13gossipAlg}, the authors proposed a broadcast-based distributed consensus method to save communication energy. Alternatively, \cite{wai:05,sch:11} describes a belief propagation/message passing approach and \cite{shu:13,lou:15,isu:15} considers signal processing on graphs. The work in \cite{Xin18directedGraph} considered distributed optimisation over a directed graph. A special class of distributed optimisation, called federated learning, focuses on collaboratively training of a machine learning model over a centralised network (i.e., a server-client topology) \cite{ Yuan20Fed, Karimireddy20SCAFFOLD}.

A method of particular interest to this work is to approach the task of distributed signal processing via its connection with convex optimisation since it has been shown that many classical signal processing problems can be recast in an equivalent convex form \cite{luo:06}. Here we model the problem at hand as a convex optimisation problem and solve the problem using standard solvers like dual ascent, method of multipliers or ADMM \cite{Boyd11ADMM} and PDMM \cite{Zhang16PDMM,Sherson17PDMM}. The solvers ADMM and PDMM, although at first sight suggested to be different due to their contrasting derivations, are closely related \cite{Sherson17PDMM}. The derivation of PDMM, however, directly leads to a distributed implementation where no direct collaboration is required between nodes during the computation of the updates. For this reason we will take the PDMM approach to derive update rules for distributed optimisation with linear equality and inequality constraints.

PDMM  was originally designed to solve the following separable convex optimisation problem 
\begin{align}
\begin{array}{ll} \text{minimise} & {\displaystyle \sum_{i\in {\cal V}} f_i(x_i)} \\\rule[4mm]{0mm}{0mm}
\text{subject to} & A_{ij}x_i + A_{ji}x_j = b_{ij}, \quad (i,j)\in \cal  E, \label{equ:linearCond}
\end{array}
\end{align}
in a synchronous setting, 
where the undirected graph $G=(\mathcal{V},\mathcal{E})$ represents a P2P network from practice. The recent work \cite{jor:23} shows theoretically that PDMM can also be implemented asynchronously, and that it is resilient to transmission losses.  In \cite{Zhang22PDMM_central}, PDMM is modified for federated learning over a centralised network, where it is found that PDMM is closely related to the SCAFFOLD \cite{Karimireddy20SCAFFOLD} and FedSplit \cite{Pathak2021} algorithm. In addition, PDMM can be used for privacy-preserving distributed optimisation where a certain amount of privacy can be guaranteed by exploiting the fact that the (synchronous) PDMM updates take place in a certain subspace so that the orthogonal complement can be used to obfuscate the local (private) data, a method referred to a subspace perturbation \cite{Jane2020ICASSP, Jane2020TIFS, Jane2020TSP, li:23}. Moreover, it has been shown in \cite{jonkman2018quantisation} that PDMM is robust against data quantisation, thereby making it a communication efficient algorithm.

For the special case of consensus problems, where
the constraints in (\ref{equ:linearCond}) are given by  $x_i=x_j$ for all $(i,j)\in \mathcal{E}$, a large number of algorithms have been proposed in the literature. Typical methods include decentralized gradient descent (DGD) \cite{Nedic09DGD}, exact first-order algorithm (EXTRA) \cite{Shi14Extra}, distributed stochastic gradient tracking \cite{Pu18GradientTracking}, and push-sum distributed dual averaging (PS-DDA) \cite{Tsianos12pushSum}.  One major difference between PDMM and the above mentioned methods is that PDMM can be derived straightforwardly by applying Peaceman-Rachford splitting, a well-known technique for decomposable optimisation. Accordingly, the convergence analysis of PDMM can be conveniently carried out by using the existing convergence theory of Peaceman-Racheford splitting (see \cite{bau:17,Sherson17PDMM} and the analysis in this paper).

\subsection{Related work}
In recent years, a number of research works (e.g., \cite{Xu19ADMMIneq, Giesen19ADMM,Chen20ADMMIneq}) have considered applying ADMM for distributed optimisation with linear inequality constraints. The basic idea is to introduce slack variables and to reformulate the inequality constraints into equality ones. The most recent work \cite{He23ADMMIneq} is an exception and tackles the linear inequality constraints differently. The authors of \cite{He23ADMMIneq} avoid introducing slack variables in extended ADMM to handle both equality and inequality constraints via a prediction-correction updating strategy. The prediction step in extended ADMM follows a similar update structure as the one in conventional ADMM and the correction step is newly introduced to ensure algorithmic convergence. In this work, we revisit PDMM for dealing with both equality and inequality constraints by applying Peaceman-Racheford splitting to the monotonic inclusion related to the lifted dual problem. Similar to \cite{He23ADMMIneq}, no slack variables are introduced in PDMM 
 to avoid any additional transmission or computation overhead between neighbours in a P2P network. The main difference between extended ADMM and PDMM is that no additional correction step is required in PDMM to handle the inequality constraints, resulting in significant faster convergence and lower computational complexity, as is demonstrated in Section~\ref{sec:exp}. The convergence of PDMM is essentially guaranteed by the convergence theory of Peaceman-Racheford splitting.  

Another related branch of work is distributed optimisation with nonlinear inequality constraints. For instance, the work \cite{Yu17parallel} proposed an effective algorithm for minimising an objective function subject to a set of nonlinear inequality constraints. The algorithm can be implemented in a parallel manner if both the objective function and the nonlinear constraints are properly decomposable. The authors of \cite{Wu21IneqCon} further extended the work of \cite{Yu17parallel} by considering additional equality constraints by combining three algorithms, where each one is designed for a particular type of constraints. %In this work, we focus on linear equality and inequality constraints in the optimisation problem by utilising  monotone operator theory, which is fundamentally different from the design principles of the above mentioned methods. 

\subsection{Main contribution}
In this work, we consider applying PDMM for distributed optimisation with both linear equality and inequality constraints. To this purpose, we make two main contributions. Firstly, to incorporate the inequality constraints, we impose nonnegativity constraints on the associated dual variables and then, inspired by \cite{Sherson17PDMM}, derive closed-form update expressions for the dual variables via Peacheman-Rachford splitting of the monotonic inclusion related to the lifted dual problem. As mentioned earlier, no additional correction step is needed in PDMM while extended ADMM in \cite{He23ADMMIneq} must introduce an additional correction step to guarantee convergence. Secondly, we perform a convergence analysis for both synchronous and stochastic PDMM. The latter is based on stochastic coordinate descent and includes asynchronous update schemes and update schemes with transmission losses.   
In addition, we give convergence conditions that are less restrictive than the ones given in \cite{Sherson17PDMM} and \cite{jor:23} for equality constrained PDMM, where strong convexity and differentiability of the objective function is assumed.

\subsection{Organisation of the paper}

The remainder of this paper is organized as follows. Section~\ref{sec:back} introduces appropriate nomenclature and reviews  properties of monotone operators and operator splitting techniques. Section~\ref{sec:pre}  describes the problem formulation while Section~\ref{sec:operator} introduces a monotone operator derivation of PDMM with inequality constraints and demonstrates its relation with ADMM. In Section~\ref{sec:convergence} we derive convergence results of the proposed algorithm and in Section~\ref{sec:stoch} we consider a stochastic updating scheme, which includes asynchronous PDMM and PDMM with transmission losses as a special case. Finally, Section~\ref{sec:exp} describes experimental results obtained by computer simulations to verify and substantiate the underlying claims of the document and the final conclusions are drawn in Section~\ref{sec:conclusion}.

\section{Background}
\label{sec:back}

There exist many algorithms for iteratively minimising a convex function.
It is possible to derive and analyse many of these algorithms in a unified manner, using the abstraction of monotone operators.
In this section we will review some properties of monotone operators and operator splitting techniques that will be used throughout this paper. For a primer on monotone operator methods, the reader is referred to the self-contained introduction and tutorial \cite{ryu:16}. For a detailed discussion on the topic the reader is referred to \cite{bau:17}.

\subsection{Notations and functional properties}

In this work we will denote by $\N$ the set of nonnegative integers, by $\R$ the set of real numbers, by $\R^n$ the set of real column vectors of length $n$ and by 
$\R^{m\times n}$ the set of $m$ by $n$ real matrices. The symbols $\succ, \succeq, \prec$ and $\preceq$ denote generalised inequality; between vectors it represents component wise inequality. We will denote by $\| x \|$ the standard Euclidean norm of $x\in\R^n$ induced by the inner product $x^Tx$. When $x$ is updated iteratively, we write $x^{(k)}$ to indicate the update of $x$ at the $k$th iteration. When we consider $x^{(k)}$ as a realisation of a random variable, the corresponding random variable will be denoted by $X^{(k)}$ (corresponding capital). The expectation operator is denoted by $\mathbb{E}$. Let ${\cal X,Y} \subseteq \R^n$. A set valued operator $T:{\cal X}\to 2^{\cal Y}$ is defined by its graph ${\rm gra}(T) = \{ (x,y)\in {\cal X}\times{\cal Y} \,|\, y = T(x)\}$, where $2^{\cal Y}$ is the power set of $\cal Y$. We define ${\rm dom}(T) = \{x\in {\cal X} \,|\, T(x) \neq \emptyset\}$. If $T(x)$ is a singleton or empty for any $x$, then $T$ is a function or single-valued,  usually denoted by $f$. 
The notion of the inverse of $T$, denoted by $T^{-1}$, is also defined through its graph, ${\rm gra}(T^{-1}) = \{ (y,x)\in {\cal Y}\times{\cal X} \,|\, y = T(x)\}$. We denote by $J_{cT} = (I + cT)^{-1}, c>0,$ the resolvent of an operator $T$ and $C_{cT} = 2J_{cT} - I$ the associated Cayley operator, sometimes referred to as the reflected resolvent.  
The composition of two operators $T_1:{\cal X}\to 2^{\cal Y}$ and $T_2:{\cal Y}\to 2^{\cal Z}$ is given by $T_2\circ T_1 : {\cal X}\to 2^{\cal Z}$.  
The set of fixed points of $T$ is denoted by ${\rm fix}(T) = \{ x\in {\cal X} \,|\, T(x)=x\}$.

Functional transforms make it possible to investigate problems from a different perspective and sometimes simplify the analysis. In convex analysis, a suitable transform is the Legendre transform, which maps a function to its Fenchel conjugate.
The Fenchel conjugate of a function $f$ is defined as $f^{\ast}(y) = \sup_x \left( y^Tx -f(x) \right)$. The function $f$ and its conjugate $f^*$ are related by the Fenchel-Young inequality $f(x) + f^*(y) \geq y^Tx$ \cite[Proposition 13.15]{bau:17}. Furthermore, the set of all closed, proper, and convex (CCP) functions $f : \R^n \to \R \cup \{+\infty\}$ is denoted by $\Gamma_0(\R^n)$ and we denote by $\partial f$ the subdifferential of $f$. If $f\in \Gamma_0(\R^n)$, then $f=f^{**}$. Moreover, we have $y \in\partial f(x) \Leftrightarrow x\in \partial f^*(y) \Leftrightarrow f(x) + f^*(y) = y^Tx$. 
If $f\in \Gamma_0(\R^n)$, the proximity operator ${\rm prox}_{cf}$ is defined as ${\rm prox}_{cf}(x) = {\arg}\min_{u\in\R^n}\left( f(u) + \frac{1}{2c}\|x-u\|^2\right)$ and is related to the resolvent of $\partial f$ by ${\rm prox}_{cf}(x) = J_{c\partial f}(x)$ \cite[Proposition 16.44]{bau:17}. If $I_C$ is the indicator function on a closed convex subset $C$ of $\R^n$, then ${\rm prox}_{I_C} = \Pi_C$, the projection operator onto $C$.

We denote an undirected graph as $G=(\mathcal{V},\mathcal{E})$, where $\mathcal{V}$ is the set of vertices representing the  nodes in the network and $\mathcal{E}=\{(i,j)\,|\, i, j\in \mathcal{V}\}$ is the set of undirected edges in the graph representing the communication links in the network. We use ${\cal E}_{\rm d}$ to denote the set of all directed edges (ordered pairs). Therefore, $|\mathcal{E}_{\rm d}|=2|\mathcal{E}|$. 
We use $\mathcal{N}_i$ to denote the set of all neighbouring nodes of node $i$, i.e., $\mathcal{N}_i=\{j \,|\,(i,j)\in \mathcal{E}\}$.  Hence, given a graph $G=(\mathcal{V},\mathcal{E})$, only neighbouring nodes are allowed to communicate with each other directly.

\subsection{Monotone operators and operator splitting}

The theory of monotone set-valued operators plays a central role in deriving iterative convex optimisation algorithms.
A prominent example of a monotone operator is the subdifferential of a convex function, and the problem at hand is expressed as finding a zero of a monotone operator (monotone inclusion problem) which, in turn, is transformed into finding a fixed point of its associated resolvent.
The fixed point is then found by the fixed point (Banach-Picard) iteration, yielding an algorithm for the original problem. 
In this section we give background information about monotone operators and operator splitting to support the remainder of this paper.
\begin{definition}[Monotone operator] 
Let $T:\R^n\to 2^{\R^n}$. Then $T$ is monotone iff for all $x,y\in {\rm dom}(T)$
\[
(T(y)-T(x))^T(y-x) \geq 0.
\] 
The operator is said to be strictly monotone iff strict inequality holds.
The operator is said to be uniformly monotone with modulus $\phi : \R_+\to [0,+\infty)$ if $\phi$ is increasing, vanishes only at 0, and
\[
(T(y)-T(x))^T(y-x) \geq \phi( \|y-x\|).
\]
The operator is said to be strongly monotone with constant $m>0$, or $m$-strongly monotone, if $T-mI$ is monotone, i.e.,
\[
(T(y)-T(x))^T(y-x) \geq m \|y-x\|^2.
\]
The operator is said to be maximal monotone iff for every $(x,u)\in \R^n\times \R^n$,
\[
(x,u)\in{\rm gra}(T) \,\, \Leftrightarrow \,\, \big(\forall (y,v)\in{\rm gra}(T)\big) \quad (v-u)^T(y-x) \geq 0.
\]
In other words, there exists no monotone operator $S:\R^n\to 2^{\R^n}$ such that ${\rm gra}(S)$ properly contains ${\rm gra}(T)$.
\end{definition}
It is clear that strong monotonicity implies uniform monotonicity, which itself implies strict monotonicity.
\begin{definition}[Nonexpansiveness]
Let $T:\R^n\to 2^{\R^n}$. Then $T$ is nonexpansive iff for all $x,y\in {\rm dom}(T)$
\[
\|T(y)-T(x)\| \leq \|y-x\|.
\]
$T$ is called strictly nonexpansive, or contractive, if strict inequality holds. The operator is firmly nonexpansive iff
for all $x,y\in {\rm dom}(T)$
\[
\|T(y)- T(x)\|^2 \leq (T(y)-T(x))^T(y-x).
\]
\end{definition}
Note that when $T$ is (firmly) nonexpansive, it is single valued and continuous.
\begin{definition}[Averaged nonexpansive operator]
Let $T : {\rm dom}(T) \to \R^n$ be nonexpansive and let $\alpha \in (0,1)$. Then $T$ is averaged with constant $\alpha$, or $\alpha$-averaged, if there exists a nonexpansive operator $S : {\rm dom}(T) \to \R^n$ such that $T = (1-\alpha)I + \alpha S$. 
\end{definition}
% Hence, if $T$ is averaged, then it is nonexpansive.
It can be shown that if $T$ is maximally monotone, then the resolvent $J_{cT}$ is firmly nonexpansive \cite[Proposition 23.8]{bau:17} and the Cayley operator $C_{cT} = 2J_{cT}-I$ is nonexpansive \cite[Corollary 23.11\,(ii)]{bau:17}. We have
\[
0\in T(x) \Leftrightarrow x \in (I+cT)(x) \Leftrightarrow (I+cT)^{-1}(x) \ni x \Leftrightarrow x=J_{cT}(x),
\] 
where the last relation holds since $J_{cT}$ is single valued. Therefore, we conclude that a monotone inclusion problem is equivalent to finding a fixed point of its associated resolvent. Moreover, since $J_{cT} = \frac{1}{2}(C_{cT}+I)$ is $1/2$-averaged, we have, by the Krasnosel'skii-Mann algorithm, that the sequence generated by the Banach-Picard iteration $x^{(k+1)} = J_{cT}(x^{(k)})$ is Fej\'{e}r monotone \cite[Definition 5.1]{bau:17} and converges weakly\footnote{In the work here we only consider finite-dimensional Hilbert spaces so that weak convergence does imply strong convergence.} to a fixed point $x^*$  of $J_{cT}$ for any $x^{(0)}\in {\rm dom}(J_{cT})$ \cite[Theorem 5.15]{bau:17}, and thus to a zero of $T$.
A prime example of this procedure is the case where $T$ is the subdifferential of a convex function. In that case the Banach-Picard iteration $x^{(k+1)} = J_{c\partial f}(x^{(k)})$ results in the well known proximal point method \cite[Theorem 23.41]{bau:17}.

For many maximal monotone operators $T$, the inversion operation needed to evaluate the resolvent may be prohibitively difficult. A more widely applicable alternative is to devise an operator splitting algorithm in which $T$ is decomposed as $T=T_1 + T_2$, and the operators $T_1$ and $T_2$ are employed in separate steps. Examples of popular splitting algorithms are the forward-backward method, Tseng's method, and Peaceman-Rachford and Douglas-Rachford splitting, where the first two methods require $T_1$ (or $T_2$) to be single valued (for example the gradient of a differentiable convex function). The Peaceman-Rachford splitting algorithm is given by the iterates  \cite[Proposition 26.13]{bau:17}
\begin{align}
x^{(k)} &= J_{cT_1}(z^{(k)}, \nonumber \\
v^{(k)} &= J_{cT_2}(2x^{(k)} - z^{(k)}), \label{eq:prs} \\
z^{(k+1)} &= z^{(k)} - 2(v^{(k)} - x^{(k)}). \nonumber
\end{align}
When $T_1$ is uniformly monotone, $x^{(k)}$ converges strongly to $x^*$ (notation $x^{(k)} \to x^*$), where $x^*$ is the solution to the monotonic inclusion problem $0 \in T_1(x) + T_2(x)$.
The iterates \eqref{eq:prs} can be compactly expressed using Cayley operators as
\begin{align*}
x^{(k)} &= J_{cT_1}(z^{(k)}),\\
z^{(k+1)} &= C_{cT_2}\circ C_{cT_1} (z^{(k)}).
\end{align*}
If either $C_{cT_1}$ or $C_{cT_2}$ 
is contractive, then $C_{cT_2}\circ C_{cT_1}$ is contractive and the Peacman-Rachford iterates converge geometrically. 
Note that since $C_{cT_2}\circ C_{cT_1}$ is nonexpansive, without the additional requirement of $T_1$ being uniformly monotone, there is no guarantee that the iterates will converge. In order to ensure convergence without imposing conditions like uniform monotonicity, we can average the nonexpansive operator. In the case of $1/2$-averaging, the $z$-update is given by
\[
z^{(k+1)}  = \frac{1}{2} \left(I + C_{cT_2}\circ C_{cT_1}\right) (z^{(k)}),
\]
which is called the Douglas-Rachford splitting algorithm.   
This method was first presented in \cite{dou:56, lio:79} and converges under more or less the most general possible conditions.
A well known instance of the Douglas-Rachford splitting algorithm is the alternating direction method of multipliers (ADMM) \cite{glo:75,gab:76,eck:93,Boy:11} or the split Bregman method \cite{bre:67}.

\section{Problem Setting}
\label{sec:pre}

To simplify the discussion, we will first consider the minimisation of a separable convex cost function  subject to a set of inequality constraints of the form $Ax\preceq b$, and later generalise this to include equality constraints as well.  That is, we first consider the following problem
\begin{equation}
\begin{array}{ll} \text{minimise} & {\displaystyle \sum_{i\in {\cal V}} f_i(x_i)} \\\rule[4mm]{0mm}{0mm}
\text{subject to} & A_{ij}x_i + A_{ji}x_j \preceq b_{ij}, \quad (i,j)\in \cal  E,
\end{array}
\label{eq:pdmmop2}
\end{equation}
where $f_i : \R^{n_i}\mapsto \R\cup\{\infty\}$ are (CCP) functions, $A_{ij}\in\R^{m_{ij}\times n_i}$ and $b_{ij}\in\R^{m_{ij}}$. 
We can compactly express \eqref{eq:pdmmop2} as
\begin{equation}
\begin{array}{ll} \text{minimise} & f(x) \\\rule[4mm]{0mm}{0mm}
\text{subject to} & Ax \preceq b,
\end{array}
\label{eq:primal}
\end{equation} 
where $x = (x_1^T,\ldots,x^T_{|{\cal V}|})^T\in \R^n$, $f(x) = \sum_{i\in {\cal V}} f_i(x_i)$, $A\in\R^{m\times n}, b\in\R^m$ with $n =\sum_i n_i$ and $m= \sum_{(i,j)} m_{ij}$. More specifically, we have $A=(a_1,\ldots, a_{|{\cal V}|}), \,a_i\in\R^{m\times n_i}$, where $a_{i}(l)=A_{ij}$  and $b(l)=b_{ij}$ if and only if $e_l = (i,j)\in {\cal E}$. Assuming the graph is connected and $m\geq n$, $A$ has full column rank.
% \richard{Without loss of generality we assume that $A$ has full rank (the constraints are linearly independent). The situation in which $A$ is rank deficient can be solved by eliminating the linearly dependent constraints defined by $A$ and $b$ and proceed with a reduced size matrix $A'$ having full rank.}
% Hence, each row of $A$ is associated to one constraint in the network. 
With (\ref{eq:primal}), the dual problem is given by
\begin{equation}
\begin{array}{ll}  \text{minimise} & f^*(-A^T\lambda) + b^T\lambda, \\\rule[4mm]{0mm}{0mm}
\text{subject to} & \lambda \succeq 0,
\end{array}
\label{eq:dual}
\end{equation}
with optimisation variable $\lambda \in\R^m$, where $\lambda =(\lambda_{ij})_{(i,j)\in\cal E}$ and $\lambda_{ij}\in\R^{m_{ij}}$ denotes the Lagrange multipliers associated to the constraints on edge $(i,j)\in \cal E$. At this point we would like to highlight that the only difference between inequality and equality constraint optimisation is that  with inequality constraint optimisation we have the additional requirement that $\lambda \succeq 0$. In the case the constraints are of the form $Ax=b$, the dual problem is simply an unconstrained optimisation problem.

\section{Operater splitting of the lifted dual function}
\label{sec:operator}

Let $A = (a_1,a_2,\ldots, a_{|\mathcal{V}|})$, where $a_i \in \mathbb{R}^{m\times n_i}$. Since
\[
f(x) = \sum_{i\in {\cal V}} f_i(x_i)  \quad\Leftrightarrow\quad f^*(y) = \sum_{i\in {\cal V}} f_i^*(y_i),
\]
that is, the conjugate function of a separable CCP function is itself separable and CCP, 
we have
\begin{equation}
f^*(-A^T\lambda) = \sum_{i\in {\cal V}} f^*_i(-a_i^T\lambda) =  \sum_{i\in {\cal V}} f^*_i\bigg(\!-\!\sum_{j\in {\cal N}_i}A_{ij}^T\lambda_{ij}\bigg).
\label{eq:optconstrij}
\end{equation}
By inspection of \eqref{eq:optconstrij} we conclude that every $\lambda_{ij}$,
associated to edge $(i,j)$, is used by two conjugate  functions: $f_i^*$ and $f_j^*$. As a consequence, all conjugate functions depend on each other. 
We therefore introduce auxiliary variables to decouple the node dependencies. That is,  
we introduce for each edge $(i,j)\in \cal E$ {two} auxiliary {node} variables $\mu_{i|j}$ and $\mu_{j|i}$, one for each node $i$ and $j$, respectively, and require that at convergence $\mu_{i|j} = \mu_{j|i} = \lambda_{ij}$. Collecting all auxiliary variables 
$\mu_{i|j}$ and $\mu_{j|i}$ into one vector $\mu\in\R^{2m}$ and introducing 
$C=(c_1,c_2,\ldots, c_{|\mathcal{V}|})$, $c_i\in\R^{2m\times n_i}$, where $c_{i}(l)=A_{ij}$ and $\mu(l)=\mu_{i|j}$ if and only if $e_l = (i,j)\in \cal E$ and $i<j$, and $c_{i}(l+m)=A_{ij}$ and $\mu(l+m)=\mu_{i|j}$ if and only if $e_l = (i,j)\in \cal E$ and $i>j$, we can reformulate the dual problem as 
\begin{equation}
\begin{array}{ll} \text{minimise} \; \;&f^*(-C^T\mu) + d^T\mu  \\
\rule[4mm]{0mm}{0mm}
\text{subject to} & \mu = P\mu,\\
& \rule[4mm]{0mm}{0mm} \mu \succeq 0,
\end{array}
\label{eq:exdual}
\end{equation}
where $C\in\R^{2m\times n}, \,d = \frac{1}{2}(b^T \; b^T)^T\in\R^{2m}$
and $P\in\R^{2m\times 2m}$ is a symmetric permutation matrix exchanging the first $m$ with
the last $m$ rows. That is, if $\eta = P\mu$, then $\eta_{i|j} = \mu_{j|i}$. We will refer to \eqref{eq:exdual} as the lifted dual problem of \eqref{eq:primal}.
Let $M = \{\mu\in\R^{2m} \,|\, \mu=P\mu, \, \mu \succeq 0\}$. Hence $M$ is closed and convex. With this, we can reformulate the dual problem as
\begin{equation}
 \text{minimise} \; \;f^*(-C^T\mu) + d^T\mu + I_M(\mu), 
\label{eq:newprob}
\end{equation}
where $I_M$ denotes the indicator function on $M$. Again, by comparing inequality vs.\ equality constraint optimisation, the difference is in the definition of the set $M$; for equality constraint optimisation the set $M$ reduces to $M = \{\mu\in\R^{2m} \,|\, \mu=P\mu\}$.
% Note that $C+PC = (A^T A^T)^T$ so that 
% \begin{equation}
% Cx + PCx = 2d \quad \Leftrightarrow \quad Ax=b.
% \label{eq:Axb}
% \end{equation}
The optimality condition for problem \eqref{eq:newprob} is given by the inclusion problem
\begin{equation}
0\in -C \partial f^*(-C^T\mu) + d + \partial I_M(\mu).
\label{eq:optconstr}
\end{equation}

In order to apply Peaceman-Rachford splitting to \eqref{eq:optconstr}, we  define  $T_1(\mu) = -C \partial f^*(-C^T\mu) + d$ and $T_2(\mu) = \partial I_M(\mu)$. To show that both operators are maximally monotone, we have
\begin{align}
( T_1&(\mu) - T_1(\eta) )^T \!(\mu-\eta) \nonumber\\
&= -\left( \partial f^*(-C^T\mu) - \partial f^*(-C^T\eta)\right)^T \! C^T(\mu-\eta) \geq 0,
\label{eq:T1}
\end{align}
since $\partial f^*$ is monotone. Similarly, 
%we find \textcolor{red}{[do we need to prove it]}
\[
(T_2(\mu) - T_2(\eta))^T\!(\mu-\eta) = \left( \partial I_M(\mu) - \partial I_M(\eta)\right)^T\! (\mu-\eta) \geq 0,
\]
and we conclude that both $T_1$ and $T_2$ are monotone. Maximality follows directly from the maximality of the subdifferential \cite[Theorem 20.25]{bau:17}.
As a consequence, Peaceman-Rachford splitting to \eqref{eq:optconstr} yields the iterates 
\begin{subequations}
\begin{align}
\mu^{(k)} &= J_{cT_1}(z^{(k)}),\\
z^{(k+1)} &= C_{cT_2}\circ C_{cT_1} (z^{(k)}).
\end{align}
\label{eq:iterates}
\end{subequations}

We will first focus on the Cayley operator $C_{cT_2}$ in (\ref{eq:iterates}), which carries the inequality constraints encapsulated by $M$. To do so, we introduce an intermediate vector $y^{(k)}$, such that 
\begin{align*}
y^{(k)} &= C_{cT_1} (z^{(k)}), \\
z^{(k+1)} &= C_{cT_2}(y^{(k)}).
\end{align*}
Since $M$ is a closed convex subset of $\R^n$, we have $J_{cT_2}(y) = {\rm prox}_{cI_M}(y) = \Pi_M(y)$, 
the projection of $y$ onto $M$. As a consequence, $C_{cT_2}$ is given by $C_{cT_2} =  2\Pi_M-I$,  the reflection with respect to $M$, which we will denote by
$R_M$. We can explicitly compute $\Pi_M(y)$, and thus $R_M(y)$.
\begin{lemma}
\[
J_{cT_2}(y) = \left[ \frac{1}{2}(I+P)y \right]^+,
\]
where $[\cdot]^+$ denotes the orthogonal projection onto the non-negative orthant. 
\label{lemma:proj}
\end{lemma}
\begin{proof}
We have
\begin{equation}
    J_{cT_2}(y) = \arg\min_{u\in M} \| u -y\|^2.
    \label{eq:minu}
\end{equation}
The corresponding Lagrangian is given by $L(u,\eta,\xi) =  \| u -y\|^2 + \eta^T(Pu-u) - \xi^Tu$. Let $\tilde{u}$ denote the optimal point of \eqref{eq:minu} and let $\tilde{\xi}$ and $\tilde{\eta}$ denote the optimal dual variables. With this, the KKT conditions are given by 
\begin{subequations}
\begin{align}
1. \,\,\,&  \tilde{u} = P\tilde{u}, \, \tilde{u} \succeq 0, \label{pf}\\
2. \,\,\,& \tilde{\xi} \succeq 0, \label{df} \\
3. \,\,\,& \tilde{\xi} \odot \tilde{u} = 0, \label{cs}\\
4.\,\,\,& 2(\tilde{u}-y) +(P-I)^T\tilde{\eta} - \tilde{\xi} = 0, \label{minL}
\end{align}
\end{subequations}
where $\odot$ denotes  component-wise multiplication. Combining \eqref{pf} and \eqref{minL} we obtain $\tilde{u} = \frac{1}{2}(I+P)y+ \frac{1}{4}(I+P) \tilde{\xi}$ so that for $\ell =1,\ldots, m: \tilde{u}_\ell = \tilde{u}_{\ell+m} = \frac{1}{2}(y_\ell + y_{\ell+m}) + \frac{1}{4}(\tilde{\xi}_\ell + \tilde{\xi}_{\ell+m})$. Hence,
if $\frac{1}{2}(y_\ell + y_{\ell+m}) > 0$, then $\tilde{u}_\ell > 0$ by \eqref{df} and thus $\tilde{\xi}_\ell=0$ by \eqref{cs}. If $\frac{1}{2}(y_\ell + y_{\ell+m})< 0$, then $\tilde{\xi}_\ell > 0$ by \eqref{pf} and thus $\tilde{u}_\ell=0$ by \eqref{cs}. If $\frac{1}{2}(y_\ell + y_{\ell+m}) =0$, then $\tilde{u}_\ell \geq 0$ by \eqref{df}. However, if $\tilde{u}_\ell > 0$, then $\tilde{\xi}_\ell=0$ by \eqref{cs}, and thus $\tilde{u}_\ell=0$, which is a contradiction. Hence $\tilde{u}_\ell=0$. This completes the proof. 
\end{proof}

Recall that $C_{cT_2} =  2\Pi_M-I = R_M$.
To get some insight in how to implement $R_M$, note that $R_M(y) = \left[(I+P)y \right]^+ - y$ where the orthogonal projection onto the non-negative orthant is due to the non-negativity constraint of $\lambda$ (and thus of $\mu$). Without this constraint, we have $J_{cT_2}(y) = \frac{1}{2}(I+P)$ and thus  $C_{cT_2} = P$, which is simply a permutation operator.  This permutation operator represents the actual data exchange in the network. That is, we have for all $(i,j)\in{\cal E} : z_{i|j} \leftarrow y_{j|i}, \, z_{j|i} \leftarrow y_{i|j}$. In the case of inequality constraints, however, we only exchange data whenever\footnote{In the case $y_{i|j}$ and $y_{j|i}$ are vector-valued, we have to do the thresholding component wise.} $y_{i|j} + y_{j|i}>0$ and locally update $z_{i|j} \leftarrow -y_{i|j}, \, z_{j|i} \leftarrow -y_{j|i}$ otherwise. Figure~\ref{fig:ref} illustrates the effect of $R_M$ for a two-dimensional example, where $\mathbf{1} = (1,\,1)^T$. If $y$ is in the halfspace $\{u : u^T\mathbf{1}>0\}$ we have $z = Py$, and $z=-y$ otherwise.
\begin{figure}[t]
\centering
\includegraphics[width=75mm]{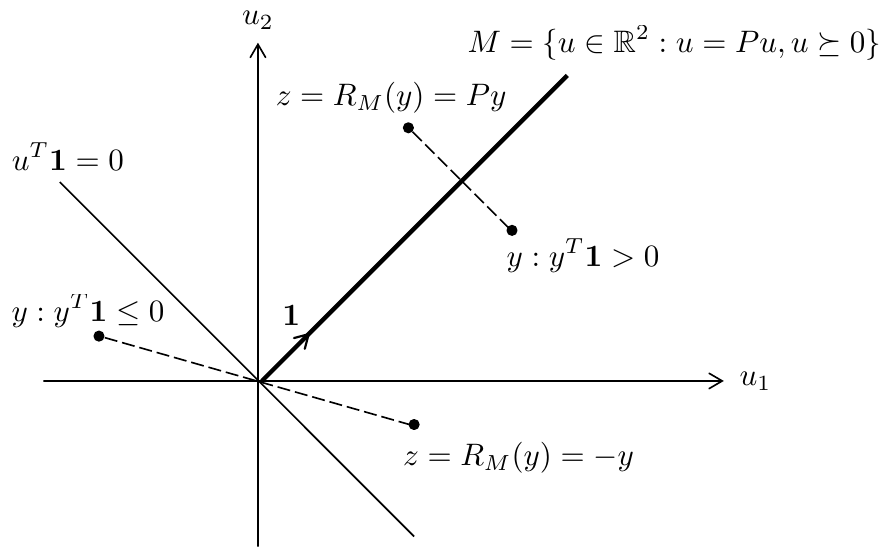}
\caption{Illustration of the reflection operator $R_M$. } 
\label{fig:ref}
\end{figure}

The iterates \eqref{eq:iterates} can now be expressed as
\begin{align*}
\mu^{(k)} &= J_{cT_1}(z^{(k)}), \\
y^{(k)} &= 2\mu^{(k)}  - z^{(k)}, \rule[4mm]{0mm}{0mm}\\
z^{(k+1)} &= R_M(y^{(k)}).\rule[4mm]{0mm}{0mm}
\end{align*}
In order to find a dual expression for $J_{cT_1}(z^{(k)})$, 
we note that
\[
\tilde{\mu}= J_{cT_1}(z) \quad \Leftrightarrow \quad z - \tilde{\mu} \in cT_1(\tilde{\mu}).
\]
Hence, $\tilde{\mu} = z +c( C\tilde{x} - d)$ where $\tilde{x} \in \partial f^*(-C^T\tilde{\mu})$, {and thus $-C^T\tilde{\mu} \in \partial f(\tilde{x})$. Hence, 
$0 \in \partial f(\tilde{x}) + C^T\tilde{\mu} =  \partial f(\tilde{x}) + C^Tz + cC^T(C\tilde{x} - d)$ so that
\[
\tilde{x} = \arg\min_x\left(f(x) + z^TCx + \frac{c}{2}\|Cx-d\|^2\right).
\]
With this, the iterates can be expressed as
\begin{subequations}
\begin{align}
x^{(k)} &=\displaystyle \arg\min_x\left(f(x) + z^{(k)\raisebox{.7mm}{\scriptsize $T$\!}}Cx + \frac{c}{2}\|Cx-d\|^2\right),  \\
\mu^{(k)} &= z^{(k)} + c(Cx^{(k)}-d), \rule[4mm]{0mm}{0mm} \label{eq:pdmmmu} \\
y^{(k)} &= 2\mu^{(k)}  - z^{(k)}, \rule[5mm]{0mm}{0mm}  \\
z^{(k+1)} &= R_M(y^{(k)}),\rule[5.5mm]{0mm}{0mm} 
\end{align}
\end{subequations}
which can be simplified to
\begin{subequations}
\label{eq:pdmm}
\begin{align}
x^{(k)} &=\displaystyle \arg\min_x\left(f(x) + z^{(k)\raisebox{.7mm}{\scriptsize $T$\!}}Cx + \frac{c}{2}\|Cx-d\|^2\right),  \label{eq:pdmma}\\
y^{(k)} &= z^{(k)} + 2c(Cx^{(k)}-d), \rule[4mm]{0mm}{0mm}  \label{eq:pdmmb} \\
z^{(k+1)} &= R_M(y^{(k)}).\rule[5.5mm]{0mm}{0mm} \label{eq:pdmmc}
\end{align}
\end{subequations}
The iterates \eqref{eq:pdmm} are collectively referred to as the {\em inequality-constraint primal-dual method of multipliers} (IEQ-PDMM).

The distributed nature of PDMM can be made explicit by exploiting the structure of $C$ and $d$
and writing out the update equations \eqref{eq:pdmm}, which is visualised in the pseudo-code of Algorithm~\ref{alg:pdmm}. 
It can be seen that no direct collaboration is required between nodes during the computation of these updates, leading to an attractive (parallel) algorithm for optimisation in practical networks. The update \eqref{eq:pdmmc} can be interpreted as one-way transmissions of the auxiliary $y$ variables to neighbouring nodes where the actual update of the $z$ variables is done.
%
% \begin{algorithm}[t]
% \vspace{-.3\baselineskip}
% \begin{align*}
% x^{(k)} &= \arg\min_x\left(f(x) + z^{(k)^T}Cx + \frac{c}{2}\|Cx-d\|^2\right) \\
% %z^{(k+1)} &= Pz^{(k)} + 2c(PCx^{(k)}-d) \rule[5mm]{0mm}{0mm} \hspace*{80mm}
% y^{(k)} &= z^{(k)} + 2c(Cx^{(k)}-d) \\
% z^{(k+1)} &= R_M(y^{(k)}) \rule[5mm]{0mm}{0mm} \hspace*{114mm}
% \end{align*}
% \vspace{-.8\baselineskip}
% \caption{Inequality-constraint PDMM}
% \label{alg:pdmm}
% \end{algorithm}
%
\begin{algorithm}[t]
\caption{Synchronous IEQ-PDMM.}\label{alg:PDMMuni}
\begin{algorithmic}[1]
\State\textbf{Initialise:}$\quad z^{(0)}\in\mathbb{R}^{2m}$ \Comment{Initialisation}
\For{$k=0,...,$}
    \For{$i\in \cal V$}\Comment{Node updates}
        \State $\displaystyle x_i^{(k)} =\arg\min_{x_i}\bigg( f_i(x_i)+$
        \Statex \hspace{2.3cm} $\sum_{j\in\mathcal{N}_i}\!\left(z_{i|j}^{(k)\raisebox{.7mm}{\scriptsize $T$\!}}A_{ij}x_i +  \frac{c}{2}\|A_{ij}x_i - \frac{1}{2}b_{ij}\|^2  \right)\!\!\bigg)$
        \ForAll{$j\in\mathcal{N}_i$}
            %\State $\boldsymbol\lambda_{i|j}^{(k+1)} = \mathbf{z}_{i|j}^{(k)}+\rho \mathbf{A}_{i|j}\mathbf{x}_i^{(k+1)}$
            \State $y_{i|j}^{(k)} = z_{i|j}^{(k)} + 2c \left(A_{ij}x_i^{(k)} - \frac{1}{2}b_{ij}\right)$
        \EndFor
    \EndFor
    \item[]
    \ForAll{$i\in{\cal V}, j\in{\cal N}_i$}\Comment{Transmit variables}
        \State $\textbf{node}_j\leftarrow\textbf{node}_i(y_{i|j}^{(k)})$
    \EndFor
    \item[]
    \ForAll{$i\in{\cal V}, j\in{\cal N}_i$}\Comment{Auxiliary updates}
        \If{$y_{i|j}^{(k)} + y_{j|i}^{(k)} > 0$}
            \State $z_{i|j}^{(k+1)} = y_{j|i}^{(k)}$
        \Else
            \State $z_{i|j}^{(k+1)} =- y_{i|j}^{(k)}$
        \EndIf
    \EndFor
\EndFor
\end{algorithmic}
\label{alg:pdmm}
\end{algorithm}

\subsection{Equality and inequality constraints}

As mentioned before the only difference in having equality or inequality constraints is in having  a nonnegativity constraint $\lambda\succeq 0$ in the latter case, and thus in the definition of the set $M$. Hence, we can trivially extend our proposed inequality constraint algorithm to include equality constraints as well. In the case of an equality constraint, we simply ignore the thresholding and exchange the associated auxiliary variables along that edge. That is, let $\mu = (\mu^T_\nu, \mu^T_\lambda)^T$, 
where $\mu_\nu$ denote the Lagrange multipliers for the equality constraints and $\mu_\lambda$ the Lagrange multipliers for the inequality constraints. Then $\mu_{\lambda} \succeq 0$, while $\mu_\nu$ is unconstrained. Defining the auxiliary variables $y$ and $z$ in a similar way, \eqref{eq:pdmmc} becomes
\begin{align*}
z^{(k+1)}_\nu &= P y^{(k)}, \\
z^{(k+1)}_\lambda &= R_M(y^{(k)}). 
\end{align*}

\subsection{Node constraints}
\label{sec:node}

In the previous sections we considered constraints of the form $A_{ij}x_i + A_{ji}x_j \preceq b_{ij}$, or $A_{ij}x_i + A_{ji}x_j = b_{ij}$ in the case of equality constraints. If we set $A_{ji}$ to be the $m_{ij}\times n_j$ zero matrix, we have constraints of the form $A_{ij}x_i \preceq b_{ij}$ or $A_{ij}x_i = b_{ij}$, which are node constraints; it sets constraints on the values  $x_i$ can take on. 
Even though $x_j$ is not involved in the constraint anymore, there is still communication needed between node $i$ and node $j$ since at the formulation of the lifted dual problem \eqref{eq:exdual} we have introduced two auxiliary variables, $\mu_{i|j}$ and $\mu_{j|i}$, one at each node, to control the constraints between node $i$ and $j$. This was done independent of the actual value of $A_{ij}$ and $A_{ji}$. In order to guarantee convergence of the algorithm, these variables need to be updated and exchanged during the iterations. Note that it is irrelevant which of the neighbouring nodes is used to define the node constraint on node $i$. We could equally well define $A_{i\ell}x_i \preceq b_{i\ell}$ with $\ell\in{\cal N}_i$, in which case there will be communication between node $i$ and node $\ell$. To avoid such communication between nodes, we can introduce dummy nodes, one for every node that has a node constraint. Let $i'$ denote the dummy node introduced to define the node constraint on node $i$. That is, we have $A_{ii'}x_i \preceq b_{ii'}$. Since dummy node $i'$ is only used to communicate with node $i$, it is a fictive node and can be incorporated in node $i$, thereby avoiding any network communication for node constraints. 

\subsection{Relation with ADMM}

Consider the prototype ADMM problem given by
\begin{equation}
\begin{array}{ll} \text{minimise} & f(x) + g(u), \\\rule[4mm]{0mm}{0mm}
\text{subject to} & Ax + Bu = c.
\end{array}
\label{eq:admm}
\end{equation} 
Following \cite{Sherson17PDMM}, we can reformulate \eqref{eq:primal} in the form \eqref{eq:admm} by introducing auxiliary variables $u_{i|j}, u_{j|i} \in \R^{m_{ij}}$ such that $u_{i|j} = A_{ij}x_i - \frac{1}{2}b_{ij}$ and $u_{j|i} = A_{ji}x_j -  \frac{1}{2}b_{ij}$. Collecting all auxiliary variables $u_{i|j}$ and $u_{j|i}$  into a vector $u\in\R^{2m}$ and using the matrices $C,P$ and $d$ as defined before, the constraints of \eqref{eq:primal} are given by $u = Cx- d$ and $u + Pu  \preceq 0$.  Hence, \eqref{eq:primal} can be equivalently expressed as
\[
\begin{array}{ll} \text{minimise} & f(x) + g(u) \\\rule[4mm]{0mm}{0mm}
\text{subject to} & Cx  - u = d,
\end{array}
\]
where $g(u)$ is the indicator function $I_{M'}$ on $M' =  \{u \in\R^{2m} \,|\, u+Pu \preceq 0\}$. The dual problem is therefore given by
\begin{equation}
 \text{minimise} \; \;f^*(-C^T\mu) + I^*_{M'}(\mu) + d^T\mu,
\label{eq:newprob2}
\end{equation}
where $\mu$, as in the PDMM case, denotes the stacked vector of dual variables $\mu_{i|j}$ and $\mu_{j|i}$ associated with the edges $(i,j)\in \cal E$.
The ADMM algorithm is equivalent to applying Douglas Rachford splitting to the dual problem  \eqref{eq:newprob2}.
Comparing \eqref{eq:newprob} and  \eqref{eq:newprob2}, we can note that the apparent difference in the dual problems is the use of
$I_M(\mu)$ in the case of PDMM and $I^*_{M'}(\mu)$ in the case of ADMM. However, we have
\[
I^*_{M'}(\mu) = \sup_u \left( \mu^Tu - I_{M'}(u)\right) = \left\{ \begin{array}{ll} 0, & \mu = P\mu, \, \mu\succeq 0\\ \infty \rule[5mm]{0mm}{0mm}& \text{otherwise,} \end{array} \right.
\]
and thus $I^*_{M'}(\mu) = I_M(\mu)$ and we conclude that the problems \eqref{eq:newprob} and  \eqref{eq:newprob2} are identical. 
As Douglas-Rachford splitting is equivalent to a half-averaged form of Peaceman-Rachford splitting, half-averaged PDMM  and ADMM will give identical results.
%\textcolor{red}{Would be nice to compare them as I did with equality constraints, but don't know how the ADMM equations as typically used should be modified.}

\section{Convergence of (in)equality-constraint PDMM}
\label{sec:convergence}

Let $T = C_{cT_2}\circ C_{cT_1}$. Since both $C_{cT_2}$ and $C_{cT_1}$ are nonexpansive, $T$ is nonexpansive, and  the sequence generated by the Banach?Picard iteration $z^{(k+1)} = T(z^{(k)})$ may fail to produce a fixed point of $T$. A simple example of this situation is $T=-I$ and $z^{(0)} \neq 0$.
Although operator averaging provides a means of ensuring algorithmic convergence, it is well known that Banach-Picard iterations converge provable faster than Krasnosel'skii-Mann iterations for the important class of quasi-contractive operators \cite{ber:04}. 
As discussed before, the Peaceman-Rachford splitting algorithm converges when $T_1$ is uniformly monotone. However, by inspection of \eqref{eq:T1}, due to the row-rank deficiency of $C$, $\exists (\mu,\eta), \mu\neq \eta : C^T(\mu-\eta) = 0$ which prohibits $T_1$ of being strictly monotone, and thus uniformly monotone. It is therefore of interest to consider if there are milder conditions under which 
certain optimality can be guaranteed. Whilst such conditions may be restrictive in the case of convergence of the auxiliary variables, in the context of distributed optimisation we are often only interested in primal optimality. For this reason we define conditions that ensure $x^{(k)}\to x^*$ even if $z^{(k)} \not\to z^*, z^*\in {\rm fix}(T)$. 
% The main crux of this proof is the observation that for uniformly convex functions $f$, the convergence of the sequence $(\|z^{(k)} - z^*\|_2 )_{k\in\N}$ will guarantee convergence of the primal variable. We have the following result.

\begin{proposition}
Let $T_1 = -C \partial f^*(-C^T(\cdot)) + d$ and $T_2 = \partial I_M$ such that ${\cal Z} = {\rm fix}(T)\neq \emptyset$ and $\partial f$ is uniformly monotone with modulus $\phi$, let $c>0$, and let $x^*$ be the solution to the primal problem \eqref{eq:primal}. Given the iterates \eqref{eq:iterates} and $z^{(0)} \in \R^{2m}$, we have $x^{(k)} \to x^*$.
\label{prop:conv}
\end{proposition}
\begin{proof}
Let $z^* \in {\cal Z}$. We have for all $k\in\N$,
\begin{align}
\|z^{(k+1)} - z^*\|^2 &= \|C_{cT_2}\circ C_{cT_1} (z^{(k)})  - C_{cT_2}\circ C_{cT_1} (z^*)\|^2 \nonumber \\
&\leq \|C_{cT_1} (z^{(k)}) -  C_{cT_1} (z^*)\|^2  \nonumber \\
&= \|2\mu^{(k)} - z^{(k)} - (2\mu^* - z^*)\|^2 \nonumber \\
&= \|z^{(k)}-z^*\|^2 \nonumber \\
&\hspace{4mm} -4(\mu^{(k)} -\mu^*)^T(z^{(k)}-\mu^{(k)} - (z^*-\mu^*)) \nonumber \\
&=  \|z^{(k)}-z^*\|^2 \nonumber \\
&\hspace{4mm}+ 4c(\mu^{(k)} -\mu^*)^TC(x^{(k)}-x^*),
\label{eq:zCT}
\end{align}
where the last equality follows from \eqref{eq:pdmmmu}. Moreover, since $x^{(k)}$ minimises $f(x) + z^{(k)^T}Cx + \frac{c}{2}\|Cx-d\|^2$, we have that
$0\in \partial f(x^{(k)}) + C^Tz^{(k)} + cC^T(Cx^{(k)}-d) =  \partial f(x^{(k)}) + C^T\mu^{(k)}$, so that
\eqref{eq:zCT} can be expressed as
\begin{align}
\|z^{(k+1)} - z^*\|^2 &\leq  \|z^{(k)}-z^*\|^2 \nonumber \\
&\hspace{4mm} - 4c(\partial f(x^{(k)}) -\partial f(x^*))^T (x^{(k)}-x^*)  \nonumber \\
&\leq  \|z^{(k)}-z^*\|^2 - 4c \phi(\|x^{(k)}-x^*\|).
\label{eq:monof}
\end{align}
Hence, $\phi(\|x^{(k)}-x^*\|)\to 0$ and, in turn, $\|x^{(k)}-x^*\|\to 0$. 
\end{proof}
\begin{remark}
Since $T$ is at best nonexpansive, the auxiliary variables will not converge in general. In fact, they will reach an alternating limit state, similar to what has been shown for equality constraint PDMM \cite{Sherson17PDMM}. In addition,
the condition for primal convergence given in Proposiiton~\ref{prop:conv} is less restrictive than the ones given in \cite{Sherson17PDMM} for equality constrained PDMM, where strong convexity and differentiability of $f$ is assumed. We will demonstrate the convergence of the algorithm for non-differentiable uniformly convex functions in Section~\ref{sec:exp}.
\end{remark}

\section{Stochastic coordinate descent}
\label{sec:stoch}
In order to obtain an asynchronous (averaged) IEQ-PDMM algorithm, we will apply randomised coordinate descent to the algorithms presented in Section~\ref{sec:operator}. 

Stochastic updates can be defined by assuming that each auxiliary variable $z_{i|j}$ can be updated based on a Bernoulli random variable $\xi_{i|j}\in\{0,1\}$. 
Collecting all random variables $\xi_{i|j}$ in the random vector $\xi\in\mathbb{R}^{2|{\cal E}|}$, following the same ordering as the entries of $z$, 
let $(\xi^{(k)})_{k\in\mathbb{N}}$ denote an i.i.d.\ random process defined on a common probability space $(\Omega, {\cal A},{\cal P})$, such that $\xi^{(k)} : (\Omega,{\cal A}) \mapsto \{0,1\}^{2|{\cal E}|}$. 
Hence, $\xi^{(k)}(\omega) \subseteq \{0,1\}^{2|{\cal E}|}$ indicates  which entries of $z^{(k)}$ will be updated at iteration $k$. We assume that the following condition holds:
\begin{equation}
    (\forall (i,j) \in {\cal E}_{\rm d}) \quad {\cal P}(\{\xi_{i|j}^{(0)} = 1\})>0.
\label{eq:xiex}
\end{equation}
Since $(\xi^{(k)})_{k\in\mathbb{N}}$ is i.i.d., \eqref{eq:xiex} guarantees that at every iteration, entry $z_{i|j}^{(k)}$ has nonzero probability to be updated. We define the block-diagonal random matrix $U\in\R^{2m\times 2m}$ as $U = {\rm diag}(\xi_{i|j} I_{m_{ij}})$.
With this,  we define the {\em stochastic} Banach-Picard iteration \cite{jor:23} as
\begin{equation}
    Z^{(k+1)} = \big(I-U^{(k)}\big)Z^{(k)}+U^{(k)}T(Z^{(k)}),
\label{eq:sbp}  
\end{equation}
where $Z^{(k)}$ denotes the random variable having realisation $z^{(k)}$.
% In Section~\ref{sec:operator} we showed that $T$ is nonexpansive. Hence, 
If $T$ is $\alpha$-averaged, a convergence proof is given in \cite{lut:13,bia:16}, where it is shown that $Z^{(k)}  - T(Z^{(k)} ) \stackrel{\rm a.s.}{\to} 0$  (almost surely).
If $T$ is not averaged, we do not have convergence in general since $T$ is at best nonexpansive and we need additional conditions for convergence. 

Let $\|z\|_Q^2 = z^T Qz$ where $Q\succ 0$ (Hermitian positive definite). Moreover, let $Q^{-1} = \mathbb{E}(U)$. Clearly, $Q \succ 0$ by condition \eqref{eq:xiex}. In addition, let $({\cal A}_k)_{k\geq 1}$ be a filtration on $(\Omega,{\cal A})$ such that
\[
{\cal A}_k := \sigma\{\xi^{(t)}: t\leq k\},
\]
the $\sigma$-algebra generated by the random vectors $\xi^{(1)},\ldots,\xi^{(k)}$ and thus ${\cal A}_k \subseteq {\cal A}_l$ for $k\leq l$. We have the following convergence result for stochastic PDMM.
\begin{proposition}
Let $T_1 = -C \partial f^*(-C^T(\cdot)) + d$ and $T_2 = \partial I_M$ such that ${\cal Z} = {\rm fix}(T)\neq \emptyset$ and $\partial f$ is uniformly monotone with modulus $\phi$, let $c>0$, and let $x^*$ be the solution to the primal problem \eqref{eq:primal}. Given the stochastic iteration \eqref{eq:sbp} and $z^{(0)} \in \R^{2m}$, we have $X^{(k)} \stackrel{\rm a.s.}{\to} x^*$.    
\label{prop:spdmm}
\end{proposition}
\begin{proof}
For any $z^* \in {\cal Z}$ we have \cite[Appendix A]{Jordan23PDMMAsyn}
\begin{align}
\mathbb{E} &\left( \|Z^{(k+1)} - z^*\|_Q^2 \, | \, {\cal A}_k \right) = \nonumber \\
&\hspace{8mm}\|Z^{(k)} - z^*\|_Q^2 +  \| T(Z^{(k)}) - z^*\|_2^2 - \|Z^{(k)} - z^*\|_2^2.
\label{eq:Eless}
\end{align}
Using \eqref{eq:monof}, \eqref{eq:Eless} becomes
\begin{align}
\mathbb{E} \left( \|Z^{(k+1)} - z^*\|_Q^2\right. & \left. \!\!| \, {\cal A}_k \right) \leq \nonumber \\
&\|Z^{(k)} - z^*\|_Q^2 - 4c\phi(\|X^{(k)}-x^*\|),
\label{eq:mart}
\end{align}
% so that $\mathbb{E} \left( \|Z^{(k+1)} - z^*\|_Q^2 \, | \, {\cal A}_k \right) \leq  \|Z^{(k)} - z^*\|_Q^2$
which shows that $(\|Z^{(k)} - z^*\|_Q^2)_{k\geq 1}$ is a nonnegative supermartingale. Moreover, since  $(\,\cdot\,)^{1/2}$ is concave and nondecreasing on $\R_+$, we conclude 
% by \cite[Theorem 25.2]{jac:04} 
that $(\|Z^{(k)} - z^*\|_Q)_{k\geq 1}$ is a nonnegative supermartingale as well and therefore converges almost surely by the martingale convergence theorem \cite[Theorem 27.1]{jac:04}. 
Taking expectations on both sides of \eqref{eq:mart} and iterating over $k$, 
% using the fact that $\mathbb{E} \left( \mathbb{E} \left( \|Z^{(k+1)} - z^*\|_Q^2 \, | \, {\cal A}_k \right) \right)  = \mathbb{E} \left( \|Z^{(k+1)} - z^*\|_Q^2 \right)$, 
we obtain 
\[
    \mathbb{E}\left(\|Z^{(k+1)} - z^*\|_Q^2\right) \leq \|z^{(0)} - z^*\|_Q^2 - 4c\sum_{t=1}^{k}\mathbb{E}\left(\phi(\|X^{(t)}-x^*\|)\right).
\]
Since $\mathbb{E}\left(\|Z^{(k)} - z^*\|_Q^2\right)\geq 0$, we have 
\begin{equation*}
    \sum_{t=1}^{k}\mathbb{E}\left(\phi(\|X^{(t)}-x^*\|)\right)\leq\frac{1}{4c}\|z^{(0)} - z^*\|_Q^2<\infty,
\end{equation*}
which shows that the sum of the expected values of the primal error is bounded. Hence, using Markov's inequality, we conclude that
\[
  \sum_{t=1}^\infty {\rm Pr}\left\{\|X^{(t)}-x^*\|^2\geq\epsilon\right\} \leq\frac{1}{\epsilon} \sum_{t=1}^\infty\mathbb{E}\left[\|X^{(t)}-x^*\|^2\right]<\infty,
\]
for all $\epsilon>0$,
and by Borel Cantelli's lemma \cite[Theorem 10.5]{jac:04}  that
\[
    {\rm Pr}\left\{\limsup_{k\rightarrow\infty}\left(\|X^{(k)}-x^*\|^2\geq\epsilon\right)\right\}=0,
\]
which shows that $\|X^{(k)}-x^*\|^2\stackrel{\rm a.s.}{\to} 0$. 
\end{proof}
\begin{remark}
As with synchronous IEQ-PDMM,
the condition for primal convergence given in Proposiiton~\ref{prop:spdmm} is less restrictive than the ones given in \cite{jor:23} for equality constrained stochastic PDMM, where strong convexity and differentiability of $f$ is assumed.
\end{remark}

\subsection{Asynchronous IEQ-PDMM}

In practice, synchronous algorithm operation implies the presence of a global clocking system between nodes. Clock synchronisation, however, in particular in large-scale heterogeneous sensor networks, can be cumbersome. In addition, due to the heterogeneous nature of the sensors/agents,  processors that are fast either because of high computing power or because of small workload per iteration, must wait for the slower processors to finish their iteration. Asynchronous algorithms
partly overcome these problems as there is much more flexibility regarding the use of the information received from other processors. Asynchronous IEQ-PDMM can be seen as a special case of stochastic IEQ-PDMM when we update a set of auxiliary variables simultaneously. That is, at each iteration, a single node, or possibly a subset of nodes chosen at random, are activated. More formally,
let $(\zeta^{(k)})_{k\in\mathbb{N}}$ denote an i.i.d.\ random process defined on a common probability space such that $\zeta^{(k)} \subseteq 2^{\cal V}$ denotes a set of indices indicating which nodes will be updated at iteration $k$. Hence, $\zeta^{(k)}$ denotes the set of {\em active nodes} at iteration $k$.
Asynchronous IEQ-PDMM can be seen as a specific case of stochastic IEQ-PDMM when we define the entries of $\xi^{(k)}$ as 
\[
(\forall (i,j) \in {\cal E}) \quad
 \xi_{j|i}^{(k)} = \left\{ \begin{array}{ll} 1 & \text{if } i \in \zeta^{(k)}, \\ 0 & \text{otherwise}. \rule[4mm]{0mm}{0mm}\end{array} \right.
\]
That is, at iteration $k$, we update all auxiliary variables $\xi^{(k)}_{i|j}, \, j\in{\cal N}_i,$ for all nodes  $i\in\zeta^{(k)}$. The pseudocode for lossy asynchronous IEQ-PDMM is given in Algorithm \ref{alg:asyn}.

\subsection{IEQ-PDMM with transmission failures}

IEQ-PDMM with transmission losses can also be seen as a special case of stochastic IEQ-PDMM. 
Let $(\eta^{(k)})_{k\in\mathbb{N}}$ denote an i.i.d.\ random process defined on a common probability space such that $\eta^{(k)} \subseteq 2^{{\cal E}_{\rm d}}$ denotes a set of ordered pairs of nodes indicating which directed edges will be updated at iteration $k$. Hence, $\eta^{(k)}$ denotes the set of {\em active directed edges} at iteration $k$; $(i,j)\in\eta^{(k)}$ implies that there has been a successful transmission from node $i$ to node $j$, but we could have a transmission failure from node $j$ to $i$.
IEQ-PDMM with transmission losses can thus be seen as a specific case of stochastic IEQ-PDMM when we define the entries of $\xi^{(k)}$ as  
\[
(\forall (i,j) \in {\cal E}_{\rm d}) \quad
 \xi_{j|i}^{(k)} = \left\{ \begin{array}{ll} 1 & \text{if } (i,j) \in \eta^{(k)}, \\ 0 & \text{otherwise}. \rule[4mm]{0mm}{0mm}\end{array} \right.
\]
Obviously, a combination of asynchronous updating and transmission loss can be modelled by defining
\[
(\forall (i,j) \in {\cal E}_{\rm d}) \quad
 \xi_{j|i}^{(k)} = \left\{ \begin{array}{ll} 1 & \text{if } i \in \zeta^{(k)} \text{ and } (i,j) \in \eta^{(k)}, \\ 0 & \text{otherwise}. \rule[4mm]{0mm}{0mm}\end{array} \right.
\]
\begin{algorithm}[t]
\caption{Asynchronous  IEQ-PDMM.}\label{alg:asyn}
\begin{algorithmic}[1]
\State\textbf{Initialise:}$\quad z^{(0)}\in\mathbb{R}^{2m}$\Comment{Initialisation}
\For{$k=0,...,$}
    \State Select a random subset of active nodes: $\zeta^{(k)}\subseteq 2^\mathcal{V}$
    % \State Select a random subset of active directed edges: $\eta^{(k)}\subseteq 2^{\mathcal{E}_{\rm d}}$
    \For{$i\in \zeta^{(k)}$}\Comment{Active node updates}
        \State $\displaystyle x_i^{(k)} =\arg\min_{x_i}\bigg( f_i(x_i)+$
        \Statex \hspace{2.3cm} $\sum_{j\in\mathcal{N}_i}\!\left(z_{i|j}^{(k)\raisebox{.7mm}{\scriptsize $T$\!}}A_{ij}x_i +  \frac{c}{2}\|A_{ij}x_i - \frac{1}{2}b_{ij}\|^2  \right)\!\!\bigg)$
        \ForAll{$j\in\mathcal{N}_i$}
            %\State $\boldsymbol\lambda_{i|j}^{(k+1)} = \mathbf{z}_{i|j}^{(k)}+\rho \mathbf{A}_{i|j}\mathbf{x}_i^{(k+1)}$
            \State $y_{i|j}^{(k)} = z_{i|j}^{(k)} + 2c \left(A_{ij}x_i^{(k)} - \frac{1}{2}b_{ij}\right)$
        \EndFor
    \EndFor
    \item[]
    \ForAll{$i\in\zeta^{(k)}, j\in{\cal N}_i$}\Comment{Transmit variables}
        \State $\textbf{node}_j\leftarrow\textbf{node}_i(y_{i|j}^{(k)})$
    \EndFor
    \item[]
    % \For{$i\in \zeta^{(k)},j\in\mathcal{N}_i :(i,j)\in \eta^{(k)}$}\Comment{Auxiliary updates}
    \For{$i\in \zeta^{(k)},j\in\mathcal{N}_i$}\Comment{Auxiliary updates}
        \If{$y_{i|j}^{(k)} + y_{j|i}^{(k)} > 0$}
            \State $z_{j|i}^{(k+1)} = y_{i|j}^{(k)}$
        \Else
            \State $z_{j|i}^{(k+1)} =- y_{j|i}^{(k)}$
        \EndIf
    \EndFor
\EndFor
\end{algorithmic}
\end{algorithm}

\section{Numerical experiments}
\label{sec:exp}

In this section we will discuss experimental results obtained by computer simulations. We will start by demonstrating that the relaxed condition ($\partial f$ being uniformly monotone) as given in Proposition~\ref{prop:conv} and Proposition~\ref{prop:spdmm} is a sufficient condition for primal convergence. We show convergence results for synchronous and asynchronous IEQ-PDMM, and demonstrate the robustness of the algorithm against transmission faillures. Secondly, we will discuss an application of network linear programming, where we collaboratively compute the intersection of convex polytopes for target localisation.  Finally, we will compare the proposed algorithm with extended ADMM \cite{He23ADMMIneq} and a PDMM variant where we introduced slack variables to handle the inequality constraints.  

\subsection{Primal convergence guarantees }
To demonstrate that PDMM doesn't converge for general problems, 
 we consider the following $\ell_1$ consensus problem:
\begin{equation}
\begin{array}{ll} \text{minimise} & \displaystyle\sum_{i=1}^n\|x_i-a_i\|_1 \\ 
\text{subject to} &   x_i=x_j, \; (i,j)\in {\cal E},
\rule[4mm]{0mm}{0mm}
\end{array}
\label{eq:l1}
\end{equation}
where the data $a_i$ was randomly generated from a Gaussian distribution.
We consider a random geometric graph of $n=25$ nodes where we have set the communication radius $r = \sqrt{2\log(n)/n}$, thereby guaranteeing a connected graph with probability at least $1-1/n^2$ \cite{dal:02}.
The resulting graph is depicted in Figure~\ref{fig:graph_25nodes}.
\begin{figure}[t]
\centering
\includegraphics[width=.345\textwidth]{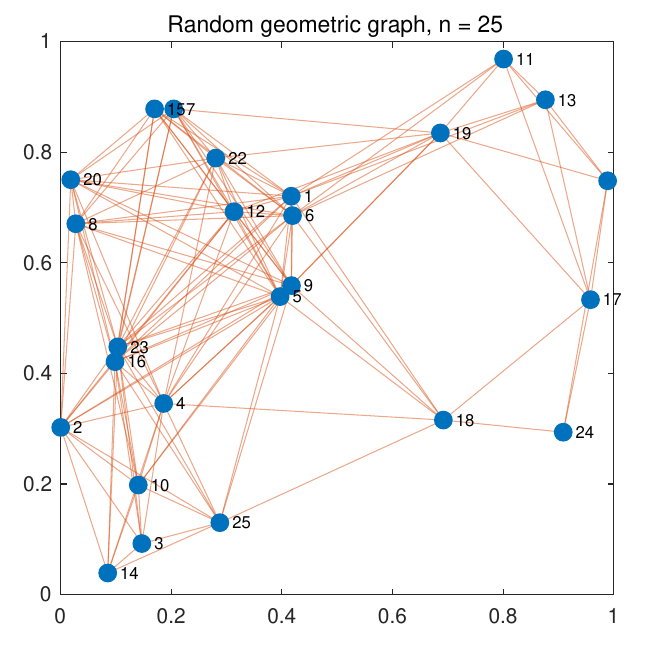}
\caption{Demonstration of a random geometric graph with 25 nodes.} 
\label{fig:graph_25nodes}
\end{figure}
\begin{figure}[t]
\centering
\hspace{-2mm}
\includegraphics[width=.45\textwidth]{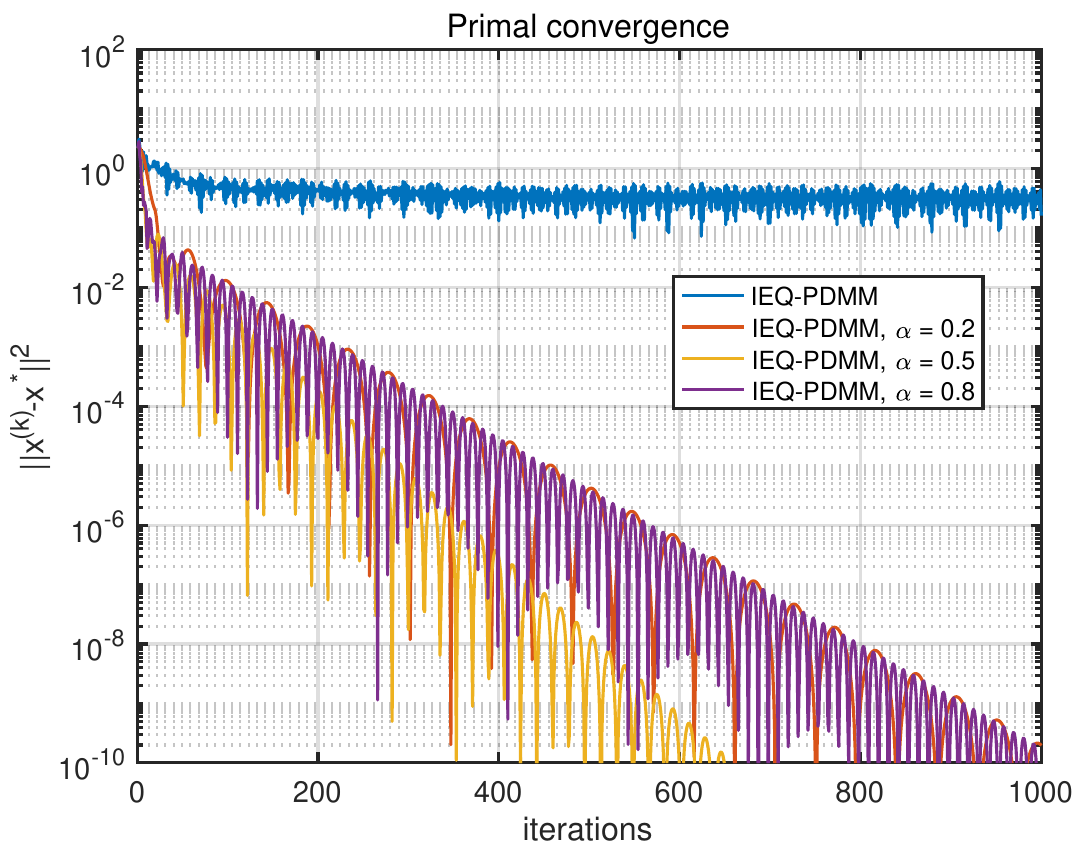}
\caption{Convergence results for IEQ-PDMM for the $\ell_1$ consensus problem \eqref{eq:l1}.} 
\label{fig:pdmm_l1}
\end{figure}
Since the objective function is not uniformly convex, the IEQ-PDMM algorithm is not guaranteed to converge. This is shown in
Figure~\ref{fig:pdmm_l1} (blue curve). In addition,  results are shown when we average IEQ-PDMM for different values of $\alpha$ in which case the algorithm is expected to converge. Note that the case $\alpha=\frac{1}{2}$ corresponds to Douglas-Rachford splitting of the lifted dual function. The step size parameter $c$ was set to $c= 0.4$.

To demonstrate that uniform monotonicity of the subdifferential $\partial f$ is sufficient for primal convergence, we consider the following extended problem:
\begin{equation}
\begin{array}{ll} \text{minimise} & \displaystyle\sum_{i=1}^n \left(\|x_i-a_i\|_1 + \|x_i-a_i\|_3^3 \right)\\ 
\text{subject to} &   x_i=x_j, \; (i,j)\in {\cal E}.
\rule[4mm]{0mm}{0mm}
\end{array}
\label{eq:l1_monotone}
\end{equation}
That is, compared to problem \eqref{eq:l1}, we added to the $\ell_1$ norm  an $\ell_3$ norm cubed, thereby making the objective function uniformly convex. Note that the resulting objective is not differentiable nor strongly convex. Figure~\ref{fig:l1_monotone} shows convergence results for problem \eqref{eq:l1_monotone}, where the MATLAB function {``fmincon''} was used as the internal optimisation solver. As expected, standard IEQ-PDMM converges for this problem and averaging slows down the convergence rate.
\begin{figure}[t]
\centering
\hspace{-2mm}
\includegraphics[width=.45\textwidth]{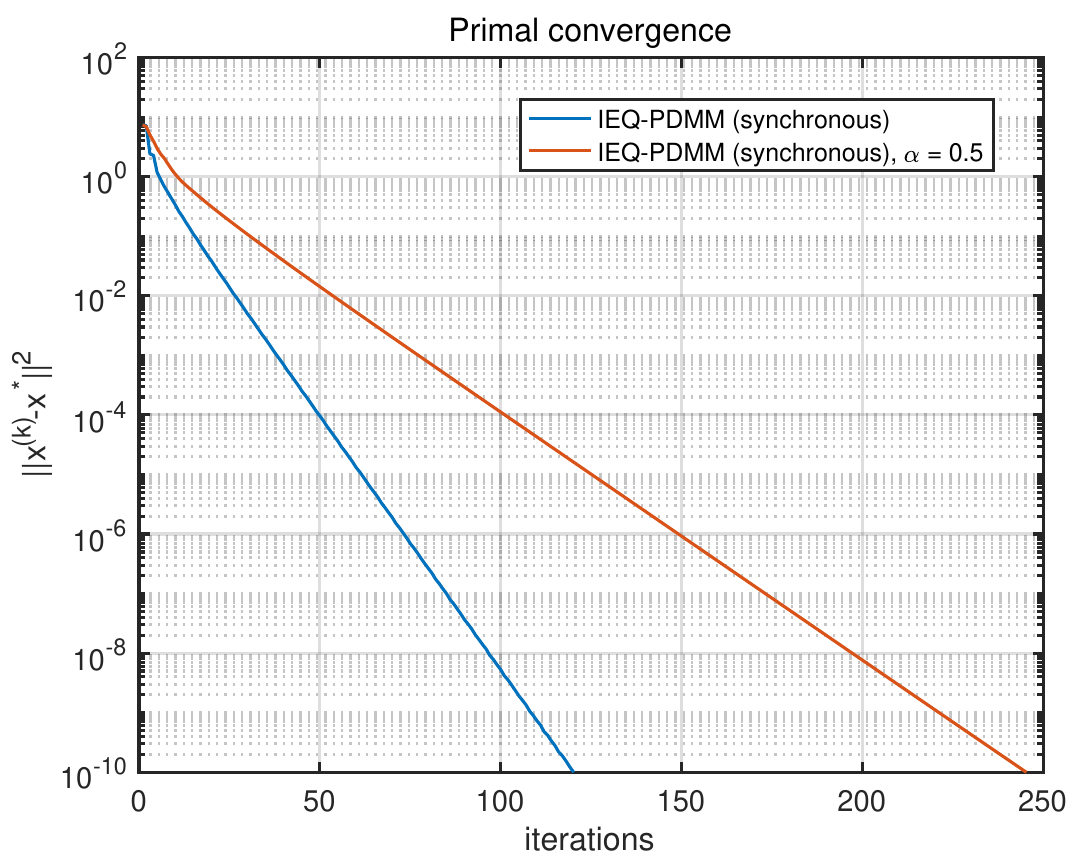}
\caption{Convergence results for IEQ-PDMM for the extended $\ell_1$ consensus problem \eqref{eq:l1_monotone}.} 
\label{fig:l1_monotone}
\end{figure}

\begin{figure}[t]
\centering
\hspace{-2mm}
\includegraphics[width=.45\textwidth]{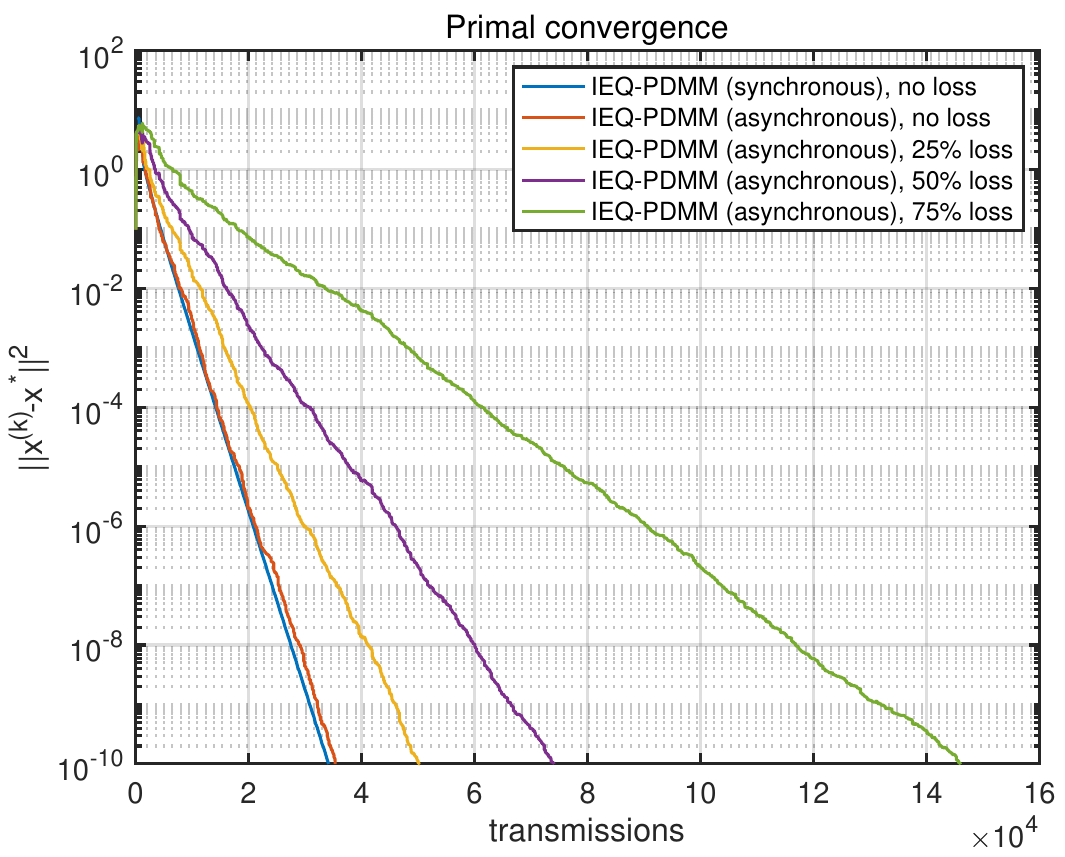}
\caption{Convergence of IEQ-PDMM for synchronous and asynchronous update schemes with different levels of transmission loss.}  
\label{fig:ieq-pdmm_stoch}
\end{figure}
To demonstrate the convergence of stochastic IEQ-PDMM, we consider problem \eqref{eq:l1_monotone} again. Figure~\ref{fig:ieq-pdmm_stoch} shows convergence results for synchronous and asynchronous IEQ-PDMM, where again the MATLAB function {``fmincon''} was used as the internal optimisation solver. The blue curve shows the result for synchronous IEQ-PDMM and is identical to the blue curve in Figure~\ref{fig:l1_monotone}. However, in order to make a meaningful comparison between synchronous and asynchronous update schemes, the convergence results are presented as a function of number of transmission rather than number of iterations. We can observe that synchronous and asynchronous IEQ-PDMM have similar convergence rates. In addition, the algorithm is robust against transmission failures and converges for all loss rates; the convergence rate decreases proportional to the loss rate, similar to what has been observed for equality constraint PDMM \cite{jor:23}.

\subsection{Target localisation}

The second simulation considers an application of network linear programming (LP) for target localisation. We consider a set of $n$ sensors randomly distributed in a unit cube which have to detect a target location $x_{\rm t}\in\R^d$. The sensors could be, for example, cameras, microphones or radars. We assume that each sensor has focused on the target by steering a beam towards the target, where we have added zero-mean Gaussian noise to the true direction to model uncertainty in the direction-of-arrival. We will model the beam as the intersection of a finite number of half-planes. In our two-dimensional example scenario, we will use two half-planes to model the beam pattern so that the sensing region is modeled by a cone. Figure~\ref{fig:target} shows such a set-up, where we have four sensors indicated by the blue dots. The dashed blue lines indicate the hyperplanes (lines in $\R^2$) modeling the sensing beams. 
\begin{figure}[t]
\centering
% \hspace{5mm}
\includegraphics[width=.38\textwidth]{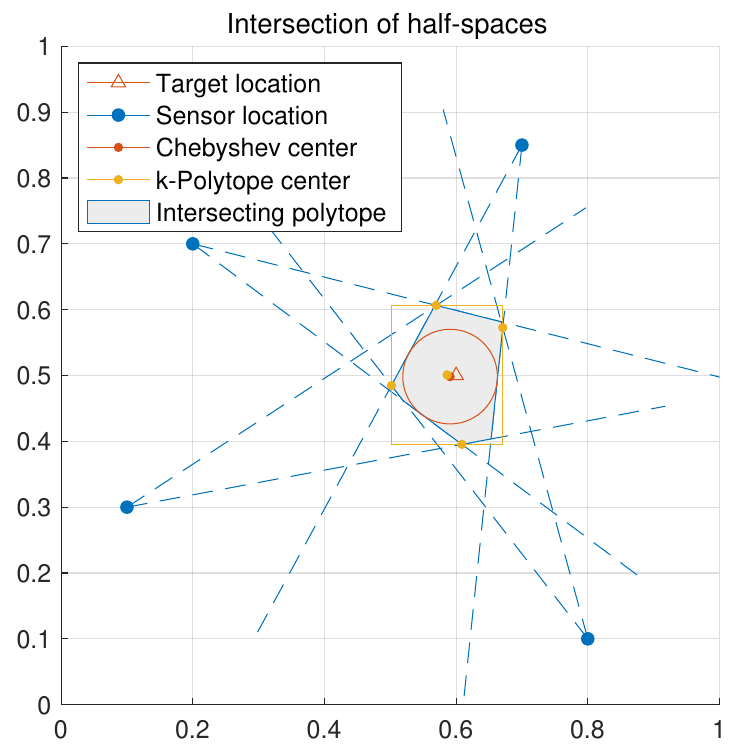}
\caption{Target location estimation by collaboratively computing the intersection of convex polytopes.} 
\label{fig:target}
\end{figure}
The intersection of the regions detected by the sensors (grey area in Figure~\ref{fig:target}) can be used to estimate the target location. Since this is the intersection of half-planes, this region is a polytope which itself is convex and non-empty. The goal is to find an inner approximation of the polytope by computing the largest Euclidean ball contained in it. The centre of the optimal ball is called the Chebyshev centre of the polytope and is the point deepest inside the polytope, i.e., farthest from the boundary. A polytope, in general, can be described as 
\[
{\cal P} = \{ x\in\R^d : a_\ell^T x \leq b_\ell, \,\ell = 1,\ldots,m\},
\]
where $m$ is the number of hyperplanes defining the polytope,
and a ball as ${\cal B} = \{ x_c + u : \|u\| < r\}$, where $x_c\in\R^d$ is the centre  and $r\in\R$ the radius of the ball. Our task is to maximise $r$ subject to the constraint ${\cal B} \subseteq {\cal P}$. Finding the Chebyshev centre can be determined by solving the LP \cite{boy:04}
\[
\begin{array}{ll} \text{maximise} & r, \\\rule[4mm]{0mm}{0mm}
\text{subject to} & a_\ell^Tx_c + r \|a_\ell\| \leq b_\ell, \; \ell=1,\ldots,m.
\end{array}
\]
In order to solve this problem distributed, we introduce local variables $x_i$ and $r_i$ at each node and add the additional constraint that $x_i=x_j$ and $r_i=r_j$ for all $(i,j)\in{\cal E}$, where ${\cal E}$ is the set of edges (communication links) in the  network. That is, we solve the LP
\begin{equation}
\begin{array}{ll} \text{maximise} & \displaystyle \sum_{i=1}^n r_i, \\\rule[4mm]{0mm}{0mm}
\text{subject to} & a_\ell^Tx_i + r_i \|a_\ell\| \leq b_\ell, \; \, i\in {\cal V}, \, \ell=1,\ldots,m, \\\rule[4mm]{0mm}{0mm}
&x_i = x_j, r_i=r_j, \; (i,j) \in {\cal E}.
\end{array}
\label{eq:cheb}
\end{equation}
Obviously, \eqref{eq:cheb} is of the form of our prototypical problem with both linear equality and inequality constraints and can, therefore,  be solved using IEQ-PDMM.
Figure~\ref{fig:target} shows the result (red circle and red triangle) for our two-dimensional example in the case of synchronous IEQ-PDMM. Figure~\ref{fig:xcconv} shows convergence results for finding the Chebyshev centre. In this example, we half-averaged the operator $T$ since the objective function is not uniformly convex and the algorithm would fail to converge without averaging.
\begin{figure}[t]
\centering
\hspace{-2mm}
\includegraphics[width=.46\textwidth]{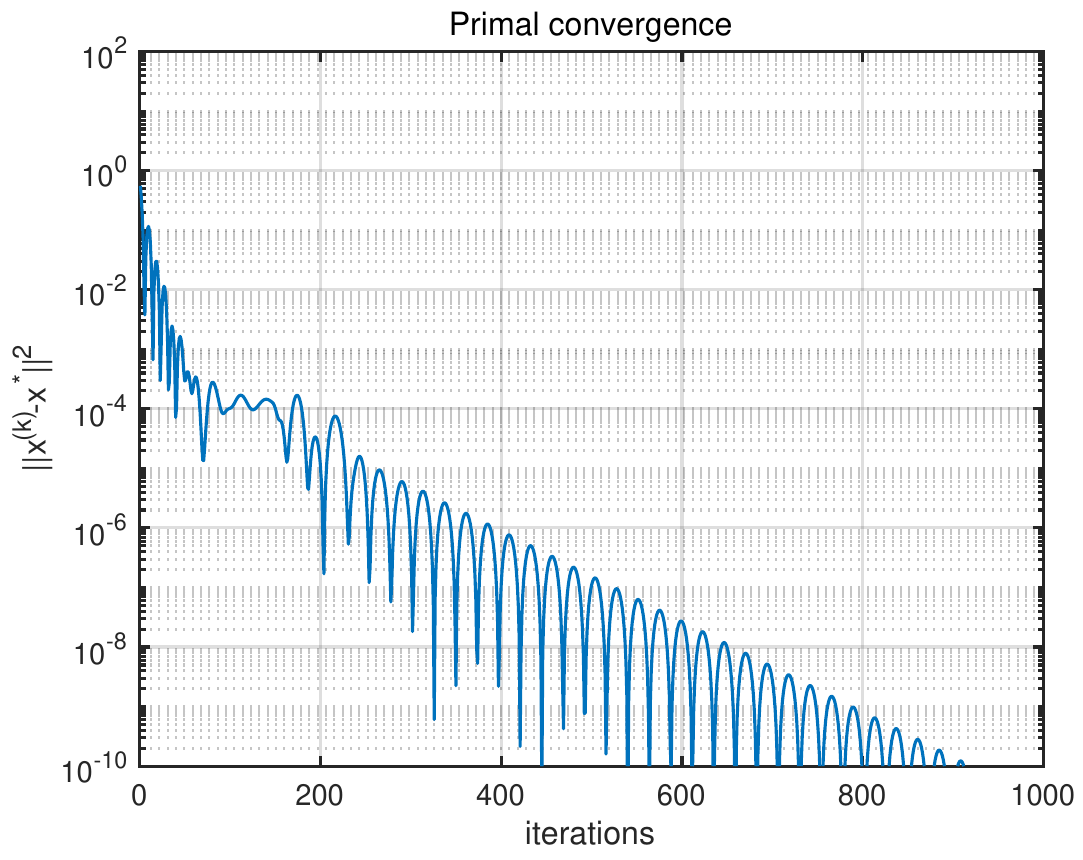}
\caption{Convergence results for finding the Chebyshev centre.} 
\label{fig:xcconv}
\end{figure}

Alternatively, we could find an outer approximation of the polytope by finding the smallest bounding rectangle (or $k$-polytope) enclosing it. In the case of a bounding rectangle, we need to solve four linear programs. Let $d_k, k=1,\ldots,4$, denote the normal vector to the $k$th hyperplane defining the bounding rectangle. We then have to solve the following LPs:
\[
\begin{array}{ll} \text{minimise} & d_k^Tx, \\\rule[4mm]{0mm}{0mm}
\text{subject to} & a_\ell^Tx \leq b_\ell, \; \ell=1,\ldots,m,
\end{array}
\]
for $k=1,\ldots,4$.
Again, in order to solver this problem distributed, we introduce local variables $x_i$ at each node and add the additional constraint that $x_i=x_j$  for all $(i,j)\in{\cal E}$. That is, we solve the LPs
\[
\begin{array}{ll} \text{minimise} & d_k^Tx_i, \\\rule[4mm]{0mm}{0mm}
\text{subject to} & a_\ell^Tx_i  \leq b_\ell, \;   i\in {\cal V}, \,\ell=1,\ldots,m,\\\rule[4mm]{0mm}{0mm}
&x_i = x_j, \; (i,j) \in {\cal E},
\end{array}
\] 
which are of the standard form suitable to be solved using IEQ-PDMM.
The result is shown in Figure~\ref{fig:target} (orange rectangle and corresponding centre point), where we again half-averaged the operator $T$ for the same reason as mentioned above.

\subsection{Comparison with existing algorithms}

In this section we consider a distributed quadratic optimisation problem with inequality constraints over the random geometric graph depicted in Figure~\ref{fig:graph_25nodes}. The problem we consider here is given by
\begin{equation}
\begin{array}{ll} \text{minimise} & {\displaystyle \sum_{i\in {\cal V}} \frac{1}{2}\|x_i-a_i\|^2} \\\rule[4mm]{0mm}{0mm}
\text{subject to} & x_i \leq  x_j \textrm{ for } i<j, \quad  \;\; (i,j)\in \cal  E,
\end{array}
\label{eq:quad_ineq}
\end{equation}
where the data $a_i$ was randomly generated from a Gaussian distribution. 
%
% \begin{figure}[t]
% \centering
% \includegraphics[width=.4\textwidth]{rand_graph_50.pdf}
% \caption{Demonstration of a random geometric graph with 50 nodes.} 
% \label{fig:graph_50nodes}
% \end{figure}
%
We compared three methods. First of all we compared the proposed IEQ-PDMM method with a PDMM variant where we introduced, as is commonly done, additional slack variables. 
The reason for this comparison is to find out if the introduction of slack variables helps 
accelerating the convergence. For every edge constraint we introduce a slack variable $w_{ij}\geq 0$ such that the inequality constraints in \eqref{eq:pdmmop2} can be expressed as 
\begin{align*}
&A_{ij}x_i + A_{ji}x_j + w_{ij} = b_{ij}, \\
&w_{ij} \geq 0.
\end{align*}
Since standard PDMM can only handle equality constraints, the inequality constraints $w_{ij}\geq 0$ can be included in the objective function by introducing the indicator function $I_{\{w\succeq 0\}}$. However, by doing so, the objective function is not separable anymore. This can be easily overcome by introducing two slack variables per edge,  $w_{i|j}\geq 0$ and $w_{j|i}\geq 0$, and add the additional equality constraint $w_{i|j} = w_{j|i}$. With this, the PDMM variant that can handle inequality constraints becomes
\[
\begin{array}{ll} 
\text{minimise} & {\displaystyle \sum_{i\in {\cal V}} \left( f_i(x_i)  + \sum_{j\in {\cal N}_i} {I}_{\{w_{i|j}\geq 0\}} \right)} \\ 
\text{subject to} & A_{ij}x_i + A_{ji}x_j+w_{i|j}+w_{j|i}=b_{ij} \rule[4mm]{0mm}{0mm} \\
& w_{i|j}-w_{j|i}=0  \rule[4mm]{0mm}{0mm} \hspace{30mm} \;\; \raisebox{0.7 em}{$(i,j)\in \cal  E$.}
\end{array} 
\]
We will refer to this algorithm as PDMM-slack. Secondly, we will compare our proposed algorithm to a state-of-the-art ADMM-based algorithm that avoids slack variables, referred to as extended ADMM \cite{He23ADMMIneq}. Since the extended ADMM algorithm is a synchronous update scheme, we only compare synchronous versions of the algorithms.
In both IEQ-PDMM and PDMM-slack, the parameter $c$ was set to $c=0.7$ and in extended ADMM \cite{He23ADMMIneq} the parameters $\nu$ and $\beta$ were set to 
$(\nu,\beta)=(0.5, 0.5)$. These values for $c, \nu$ and $\beta$ were chosen to produce the fastest convergence rate.
% The two hyper-parameters $\nu$ and $\beta$ in extended ADMM \cite{He23ADMMIneq} were searched from the discrete set $\{0.1, 0.2,\ldots, 1.0\}$ to produce the fastest convergence speed. The optimal setup after searching is $(\nu^{\ast},\beta^{\ast})=(0.5, 0.5)$. 
To make a fair comparison between the methods, we set $\alpha = \frac{1}{2}$ (ADMM). 
For completeness, we included results for $\alpha = 1$ as well. 
\begin{figure}[t]
\centering
\hspace*{-5mm}
\includegraphics[width=.43\textwidth]{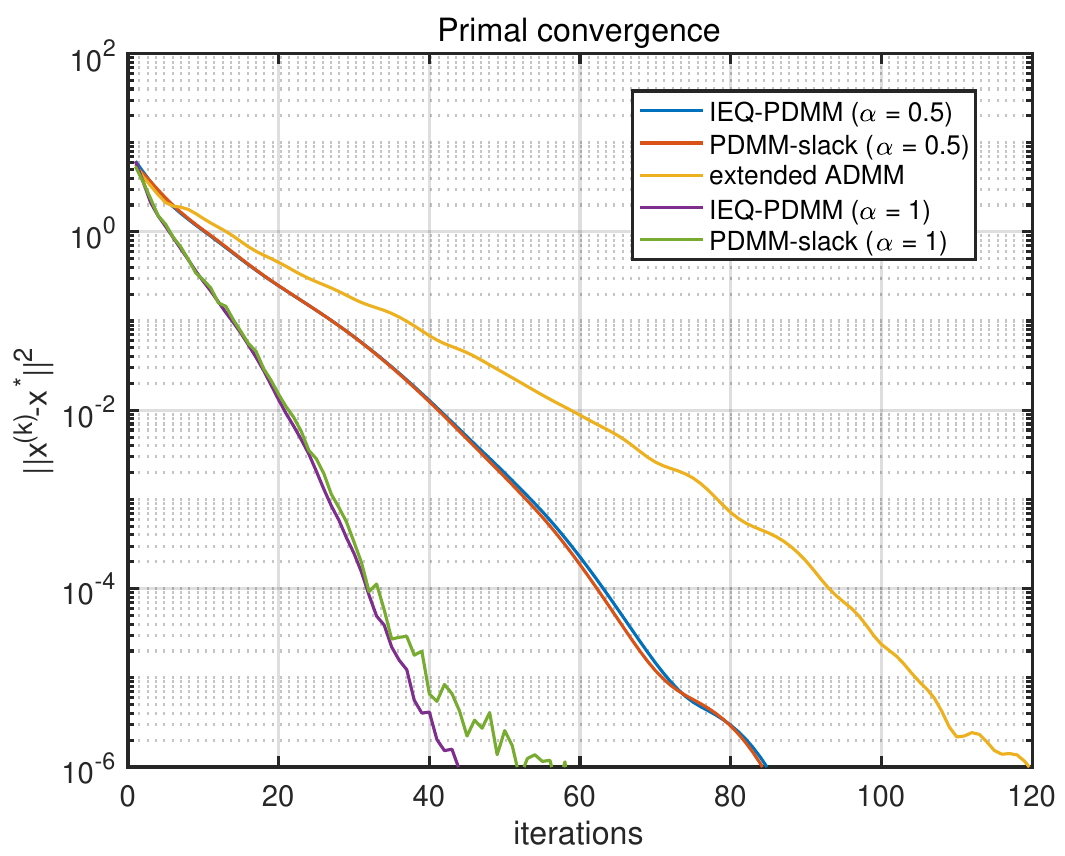}
\caption{Convergence comparison of IEQ-PDMM, PDMM-slack, and extended ADMM over the random geometric graph depicted in Figure~\ref{fig:graph_25nodes}.} 
\label{fig:syn_compare}
\end{figure}
Figure~\ref{fig:syn_compare} visualises the convergence results of the three methods. As can be seen, both PDMM algorithms have similar convergence rates and outperform the extended ADMM algorithm in terms of number of iterations needed to converge to a certain accuracy level. However, the computational complexity of the proposed IEQ-PDMM algorithm is significantly lower than the extended ADMM and PDMM-slack algorithm. This can be seen from Table~\ref{table:time_complexity} which shows the average time (in seconds) needed per iteration for the three methods. Clearly, IEQ-PDMM consumes the least amount of time, demonstrating its efficiency. The PDMM-slack algorithm is most expensive because we have to perform an inequality constraint optimisation problem at each and every iteration (here implemented using the MATLAB program ``quadprog''). The above results indicate that the introduction of slack variables does not improve the convergence rate of PDMM and that it is most efficient to handle the inequality constraints directly by imposing non-negativity constraints on the dual variables as is done in (\ref{eq:dual}).
\begin{table}[t!]
\caption{\small Comparison of computational complexity per iteration (in seconds) } 
\label{table:time_complexity}
\centering
\begin{tabular}{|c|c|c|}
\hline 
 \scriptsize{extended ADMM \cite{He23ADMMIneq}} & \scriptsize{IEQ-PDMM } & \scriptsize{PDMM-slack }  \\
\hline
 % \scriptsize{0.0032} &  \scriptsize{0.2076} & \scriptsize{0.0479}  \\
 \scriptsize{$4.5\cdot 10^{-4}$}  & \scriptsize{$3.9\cdot 10^{-5}$} &  \scriptsize{$8.6\cdot 10^{-2}$}
 \rule[3mm]{0mm}{0mm}\\
\hline
\end{tabular}
\vspace{-.5\baselineskip}
\end{table}

\section{Conclusions}
\label{sec:conclusion}

In this paper we have presented a node-based distributed optimisation algorithm for optimising a separable convex cost function with linear equality and ineqaulity node and edge constraints, termed inequality-constraint primal-dual method of multipliers (IEQ-PDMM). Using monotone operator theory and operator splitting, we derived node-based update rules for solving the problem. To incorporate the inequality constraints, we imposed non-negativity constraints on the associated dual variables, resulting in the introduction of a reflection operator to model the data exchange in the network, instead of a permutation
operator as derived for equality constraint PDMM. We showed how to avoid unnecessary communication between nodes in the case we have node constraints by introducing fictive nodes in the network and highlighted the relation with Peaceman-Rachford splitting and ADMM. We showed convergence results for both synchronous and stochastic update schemes, where the
latter includes asynchronous update schemes and update schemes with transmission losses. The algorithm converges for any CCP cost function when using averaged iterations, and has primal convergence for non-averaged updates in the case the cost function is uniformly convex. 

%\bibliographystyle{IEEEtran}
%\bibliography{sigProcessing_sig,sigProcessing_sig_2nd}
%


\end{document}